\documentclass{article}

\usepackage{arxiv}

\usepackage[utf8]{inputenc}
\usepackage[T1]{fontenc}    
\usepackage{hyperref}       
\usepackage{url}            
\usepackage{booktabs}       
\usepackage{amsfonts}       

\usepackage{nicefrac}       
\usepackage{microtype}                   
\usepackage{graphicx}
\usepackage{verbatim}

\usepackage{amsmath,amssymb,amsthm}
\newtheorem{proposition}{Proposition}

\newtheorem{theorem}{Theorem}
\newtheorem{definition}{Definition}[section]

\usepackage{tikz}
\usetikzlibrary{calc}
\usetikzlibrary{shapes.geometric, positioning, arrows.meta}
\usetikzlibrary{matrix,calc}
\usetikzlibrary{fit,decorations.pathreplacing}

\usepackage{multirow}

\usepackage{enumitem}
\setlist[itemize]{labelindent=1em,leftmargin=2em,labelsep=0.5em}

\usepackage{braket}
\usepackage[english]{babel}

\usepackage{array}
\usepackage{tabularx}

\newcolumntype{L}[1]{>{\raggedright\arraybackslash}p{#1}}
\newcolumntype{Y}{>{\raggedright\arraybackslash}X}

\makeatletter
\def\blfootnote{\xdef\@thefnmark{}\@footnotetext}
\makeatother

\usepackage{fancyhdr}
\begin{document}

\title{Encrypted Redundancy as a Diagnostic Resource: Relational Diagnosis in Quantum Encrypted Cloning}

\author{
  \href{https://orcid.org/0000-0001-5186-0199}{\includegraphics[scale=0.06]{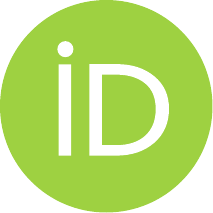}\hspace{1mm} Gabriele Gianini}\\
  University of Milano-Bicocca \\
  Milan, Italy \\
  \texttt{gabriele.gianini@unimib.it}\\ \vspace{1pt}
\and
\href{https://orcid.org/0000-0002-9585-7810}{\includegraphics[scale=0.06]{orcid.pdf}\hspace{1mm}
  Omar Hasan}\\
  INSA of Lyon\\
  Lyon, France \\
  \texttt{omar.hasan@liris.cnrs.fr}\\ \vspace{1pt}
\and
  \href{https://orcid.org/0000-0003-1737-6218}{\includegraphics[scale=0.06]{orcid.pdf}\hspace{1mm}
  Stelvio Cimato}\\
  University of Milano\\
  Milan, Italy\\
  \texttt{stelvio.cimato@unimi.it}\\ \vspace{1pt}
\and
  \href{https://orcid.org/0000-0002-9557-6496}{\includegraphics[scale=0.06]{orcid.pdf}\hspace{1mm}
  Ernesto Damiani}\\
  University of Milano\\
  Milan, Italy\\
  \texttt{ernesto.damiani@unimi.it}
}
\date{}
\maketitle

\begin{abstract}
Quantum encrypted cloning encodes an unknown state into several
encrypted components, each offering an alternative way to recover it later. We show that this redundancy can also serve for one-shot fault
diagnosis: instead of inspecting the encrypted state, we measure relational
Pauli observables testing consistency conditions imposed by the encoding. The canonical Yamaguchi--Kempf scheme encodes the input into \(n\) signal--key pairs, all carrying the same coherent Bell label, plus a transformed copy of the input qubit. Comparing the labels of two pairs yields a deterministic check that localizes an anomalous pair and identifies its Pauli-error class without revealing the label. Such checks cannot tell whether the signal or the key of that pair is faulty, and no measurement on the pairs can (we prove that they already generate the whole group of deterministic state-blind observables
supported there, of rank \(2(n-1)\)). This ambiguity matters: a faulty signal costs one redemption path, a faulty key threatens them all. 
However, retaining also the transformed input qubit contributes exactly two
further independent checks, at any multiplicity, and these suffice to identify
any single fault drawn from the single-qubit Pauli set. For three clones this gives six checks in all, and we prove that no smaller set
of state-blind observables achieves the same resolution. Across clone multiplicities, parity governs how many components must be measured jointly to reach this resolution, not the resolution itself.
These findings are then abstracted into a general framework --- state-blind
observables, deterministic healthy references, syndrome-induced fault
partitions, redemption-oriented sufficiency --- yielding \textit{relational
diagnosability}: a syndrome need not identify every fault, only enough of it to select a safe redemption path.
\keywords{
Quantum encrypted cloning;
encrypted quantum redundancy;
relational diagnosability;
one-shot relational diagnosis;
fault detection and localization;
Pauli-string measurements;
redemption paths;
quantum secret sharing;
}
\end{abstract}

\newpage
\tableofcontents

\section{Introduction}
\label{sec:intro}

An unknown quantum state cannot be copied in the ordinary sense
\cite{wootters1982single,dieks1982communication,barnum1996noncommuting}; Approximate and probabilistic cloning relax but
do not remove these restrictions
\cite{buzek1996quantum,duan1998probabilistic}. Consequently, classical
reliability mechanisms based on the availability of several simultaneously
readable replicas, such as replication and majority voting
\cite{patterson1988raid,reich2006lockss}, cannot be transferred unchanged to
quantum storage and custody.

Quantum encrypted cloning (QECL) provides a different form of redundancy by
changing what constitutes a ``copy.'' Rather than creating several readable
instances of an unknown state, QECL encodes the state into encrypted quantum
components that provide alternative opportunities for later reconstruction
\cite{yamaguchi2026encrypted,yamaguchi2026experimental,
ceara2026cloningencryptedquantumstates}. A suitable quantum decryption key can
be combined with one of several encrypted components to recover the protected
state. The construction remains compatible with the no-cloning theorem because
these alternatives are mutually exclusive: redeeming one clone consumes the
relevant quantum resource and prevents simultaneous recovery of independent
copies.

This recovery-oriented perspective naturally connects QECL with quantum secret
sharing (QSS), where an access structure specifies which subsets of distributed
shares can reconstruct a secret and which remain partially or completely
uninformative
\cite{hillery1999quantum,cleve1999share,gottesman2000theory,
ogawa2005quantum,lim2026encrypted}. More specifically, suitable QSS access
structures admit an \textit{operational reinterpretation} as QECL schemes \cite{gianini2026beyond}. Consider a
family of qualified sets \(Q_1,\ldots,Q_m\) with a nonempty common intersection
\(K\), where \(K\) alone cannot reconstruct the secret. Writing
\(C_i=Q_i\setminus K\), each qualified set takes the form
\(Q_i=K\cup C_i\). The common component \(K\) can then be interpreted as a
quantum decryption key and each residual component \(C_i\) as an encrypted
clone: \(C_i\) need not reveal the unknown state by itself, whereas
\(K\cup C_i\) constitutes an alternative recovery opportunity. Under this
interpretation, QSS \textit{access structures} become QECL \textit{redemption structures}\footnote{We use \emph{redemption} rather than \emph{recovery} throughout, because the two differ in a way that matters here. Classical replicas can be read repeatedly and independently; a QECL clone cannot. Redeeming one encrypted clone consumes the quantum resource that made it usable and
forecloses the alternatives, so the available redemption sets are mutually exclusive opportunities rather than interchangeable copies. The term
is borrowed from the redemption of a claim against a single underlying asset,
and it is what makes the diagnostic question of this paper worth asking: since one cannot try a path and fall back on another, the choice must be informed before it is made.} and
different overlap patterns induce different key--clone organizations.

\paragraph{The diagnostic question.} When such encrypted redundancy is deployed in a distributed storage or custody
setting~\cite{gianini2026distributed,gianini2026aftertheft}, however, a further operational question
arises: \emph{before selecting a redemption path, can one determine which of
the available encrypted resources remain trustworthy without accessing or
destroying the unknown quantum information?}

This question differs fundamentally from ordinary state verification. The
protected state is unknown, and individual encrypted components may
intentionally reveal no usable information about it. 
At the same time, the distributed components are not arbitrary independent
systems: they originate from a common coherent encoding and may consequently
\textit{satisfy relations that are independent of the logical payload}. Such relations provide a possible diagnostic reference even when the individual encrypted states themselves cannot be inspected.

This observation motivates a second use of encrypted redundancy. In addition
to providing \textit{alternative redemption paths}, the encoding may provide
\emph{diagnostic redundancy}: different encrypted components can act as mutual
witnesses of whether the joint resource still satisfies the constraints imposed
by the QECL construction. The resulting diagnosis is relational rather than
absolute. Instead of asking whether an individual component occupies a known
quantum state, we ask whether selected components continue to \textit{satisfy
state-independent relations expected of every valid encoding}.

Notice that a measurement may be blind to the logical state and yet have intrinsically random outcomes in the healthy encoding; such a measurement
provides no definite single-shot reference. We investigate diagnostic checks that are state blind but also have a \textit{deterministic healthy outcome} fixed by the encoding itself: faults can then be detected through departures from that relation without resolving the protected payload.

\paragraph{One-shot diagnosis and its relation to error correction.} The present work focuses on \emph{one-shot relational diagnosis}: the
diagnostic procedure acts on the currently available encoded resource without
requiring a previously recorded syndrome for the same unknown state or
repeated identical preparations. 

The operational objective of this investigation differs from that of \textit{conventional quantum error correction} (QEC). A QEC syndrome is normally used to infer an error and restore an encoded logical block. In QECL, complete physical error identification and correction are not always necessary. If the available diagnostic information
establishes that a path-specific encrypted clone is unreliable while another
redemption path remains viable, the affected path may simply be excluded.
Active correction becomes necessary when faults affect resources shared among
several redemption paths, when no suitable alternative remains, or when
preserving additional recovery opportunities is itself an objective. The
relevant diagnostic question is therefore not necessarily ``which physical
fault occurred?'' but rather \emph{has enough been learned about the faults to make a safe redemption decision?}

This objective is closer to \textit{classical system-level fault
diagnosis}~\cite{preparata1967connection}, where a faulty unit need only be confined to a bounded set~\cite{barsi1976theory}, and to \textit{error-locating
codes}~\cite{wolf1963errorlocating}, where the sub-block containing an error is identified but not corrected, than to conventional error correction; the
correspondence is examined in Section~\ref{subsec:classical-diagnosis}.

\paragraph{Contributions.} The contributions of this work are the following.
\begin{itemize}
\item
We formulate \emph{one-shot relational diagnosis} as a distinct use of
encrypted quantum redundancy and develop a general framework based on
state-blind observables, deterministic healthy references, syndrome-induced
fault equivalence classes, and redemption-oriented diagnostic sufficiency. This
leads to the notion of \emph{relational diagnosability} as a property of a
QECL protocol relative to its fault model, available measurements, and
redemption structure.

\item
We expose the basic diagnostic mechanisms through a prototypical duplicated-syndrome QECL
construction. 
\begin{itemize}
\item
With this we also illustrate the distinction between \textit{undetected harmful} and
\textit{redemption-irrelevant} errors.
\end{itemize}
\item
We derive complementary diagnostic mechanisms for the canonical
Yamaguchi--Kempf protocol and determine their limits. That protocol encodes the
input qubit into \(n\) signal--key pairs \(S_iK_i\), all carrying the same
coherently repeated Bell label, together with a transformed copy \(A\) of the
input itself; normal redemption uses one signal together with the complete key
block. 
\begin{itemize}
\item
Inter-pair checks compare the Bell labels carried by two different pairs
and, for arbitrary clone multiplicity, localize an anomalous pair and identify
the corresponding Pauli-error class. What such a comparison cannot do is say which member of the pair (the path-specific signal or the shared key) actually failed. We prove that this is not a shortcoming of the particular observables chosen but a structural limit: the inter-pair checks already generate the \emph{entire} group of deterministic state-blind observables supported on the \(n\) pairs, a group of rank \(2(n-1)\), so the signal--key ambiguity cannot be removed by any further measurement confined to the stored pairs. 
\item
This limitation can be removed only by involving \(A\), and exactly two independent checks do so, at every clone multiplicity; adding them resolves every single-qubit Pauli fault uniquely.
\item
For three clones this gives a six-generator syndrome, which we show to be
minimal: no five deterministic state-blind checks separate all single-qubit
Pauli faults, even though a counting argument would permit five. We also
characterize explicitly the residual equivalence classes of the cheaper
five-generator variant, and observe that a five-check set of strictly finer
partition can nonetheless be operationally worse.
\end{itemize}
\item
We clarify the relation between \textit{confidentiality} and \textit{diagnosis}. The odd--even
comparison connects \textit{relational diagnosability} with previously characterized
\textit{leakage properties}~\cite{gianini2026encrypted,gianini2026full,bai2026classification}: for odd clone
multiplicity a restricted leakage channel coexists with a low-weight
state-independent relation useful for diagnosis, and eliminating the former
also eliminates the latter. We show, however, that this is a trade-off against
the \emph{cost} of diagnosis rather than against diagnosability itself, since
complete single-fault resolution remains attainable at either parity, at higher
measurement weight. \textit{Secrecy}, \textit{recovery redundancy}, and \textit{diagnosability} are thus
distinct properties, and comparing protocols along the diagnostic axis requires
a cost model as well as a resolution criterion.
\end{itemize}

\paragraph{Structure of the paper.} Section~\ref{sec:distributed-custody} introduces the distributed-custody setting and the diagnostic hierarchy used throughout. Section~\ref{sec:prototype} illustrates the basic mechanism with a duplicated-syndrome construction, and Section~\ref{sec:canonical-protocol} develops the canonical Yamaguchi--Kempf case, including storage-only and source-assisted checks, their completeness and minimality properties, and the dependence of measurement cost on clone parity. Section~\ref{sec:relational-diagnosis} abstracts these results into a general framework for one-shot relational diagnosis; Section~\ref{sec:discussion} discusses their implications and limitations, and Section~\ref{sec:conclusions} concludes.

\paragraph{Scope.} Our focus is deliberately on the \textit{logical structure of one-shot diagnosis}
rather than on the optimization of its physical implementation. The Pauli-string
observables considered here can in principle be measured through
ancilla-assisted local or distributed constructions, but their resource-aware
compilation, fault-tolerant extraction, and integration with optimized local
quantum error correction constitute separate problems.

\section{Quantum Encrypted Cloning as a Distributed Custody Primitive}
\label{sec:distributed-custody}
This section fixes the distributed-custody setting needed for the diagnostic analysis: alternative redemption paths, their physical separation, and the hierarchy in which protocol-specific relational information sits.


\subsection{Encrypted Redundancy and Redemption Paths}
\label{subsec:redemption-paths}

Let \(V\) denote a QECL encoding that maps an unknown input state into a multipartite quantum system \(R\). The encoded information is not intended to be recovered from every subsystem of \(R\), but only from selected qualified subsets. We denote by
\(
\Gamma_{\mathrm{red}} = \{Q_1,\ldots,Q_m\}
\)
the family of redemption sets, where each \(Q_i \subseteq R\) supports a recovery map \(\mathcal{D}_i\) satisfying
\[
\mathcal{D}_i\!\left(\operatorname{Tr}_{R \setminus Q_i}
\left[V \rho V^\dagger\right]\right) = \rho
\]

for every input state \(\rho\). The elements of \(\Gamma_{\mathrm{red}}\) therefore represent alternative future opportunities for recovering the same unknown quantum information.

In the simplest encrypted-cloning architectures, the redemption sets share a common quantum key \(K\) and differ only in the encrypted clone that completes the recovery configuration,
\[
Q_i = K \cup C_i .
\]
This form separates path-specific resources from shared ones. A fault on a clone \(C_i\) may remove only \(Q_i\), whereas a fault on the shared key \(K\) can affect every redemption path containing it. More generally, a subsystem is operationally critical according to the redemption sets in which it appears.

The diagnostic objective is therefore redemption-oriented: the relevant question is not whether every component remains in its ideal physical state, but whether the available information is sufficient to identify viable elements of \(\Gamma_{\mathrm{red}}\), exclude unsafe ones, or determine when correction is required.

\subsection{Distribution and Independent Protection}
\label{subsec:distribution-protection}

Once the QECL encoding has created the family of redemption resources
associated with \(\Gamma_{\mathrm{red}}\), the corresponding subsystems can be
physically separated and assigned to distinct storage sites. Physical
separation is the mechanism by which the \textit{logical redundancy}
represented by \(\Gamma_{\mathrm{red}}\) is converted into \textit{operational
redundancy}. Let \(R=R_1\cdots R_N\) be the encoded subsystems and let
\(\ell(R_j)\) denote the site holding \(R_j\). Writing
\(\ell(C_i)=\{\ell(R_j):R_j\in C_i\}\) for the sites of the path-specific
components of \(Q_i=K\cup C_i\), the natural requirement on the distribution
map is that no single failure domain intersect \(\ell(C_i)\) for every \(i\),
so that no localized outage removes all recovery opportunities at once. The
shared key block \(K\) escapes this requirement by construction, since it
belongs to every redemption set.

Accordingly, physical protection need not be uniform across components: the same local error rate can have different consequences depending on the component's role in the redemption structure. A broader treatment of transmission, storage, protection, and redemption strategies is given in~\cite{gianini2026distributed}.

Independence of failure domains, moreover, concerns the mechanisms of failure
and not the consequences of a failure once it has occurred. The components are
not informationally independent: by construction they retain the
\textit{state-independent quantum relations} that will later be interrogated,
and a fault confined to one site, though local in origin, is non-local in
effect, since it displaces the joint resource from a relation that every valid
encoding satisfies. Two caveats follow. Common-mode or correlated disturbances
may affect several components at once, and relational checks are more sensitive
to differential faults than to transformations acting coherently on all of
them. And the joint measurements that extract a relational syndrome are
themselves operations, able to propagate errors between the components they
couple. We assume throughout ideal syndrome extraction and a single-qubit Pauli
fault model; both idealizations are discussed in
Section~\ref{subsec:limitations}.

After distribution the system consists of spatially separated, independently
maintained resources whose collective usefulness is determined by the surviving
redemption sets. Diagnostic information may then be produced locally, through
device telemetry or quantum-error-correction syndromes, or relationally,
through measurements involving components at different sites. Physical
separation is introduced to obtain independent protection, yet the encoded
components retain relations that can still be interrogated across it; that
possibility is the subject of this work.


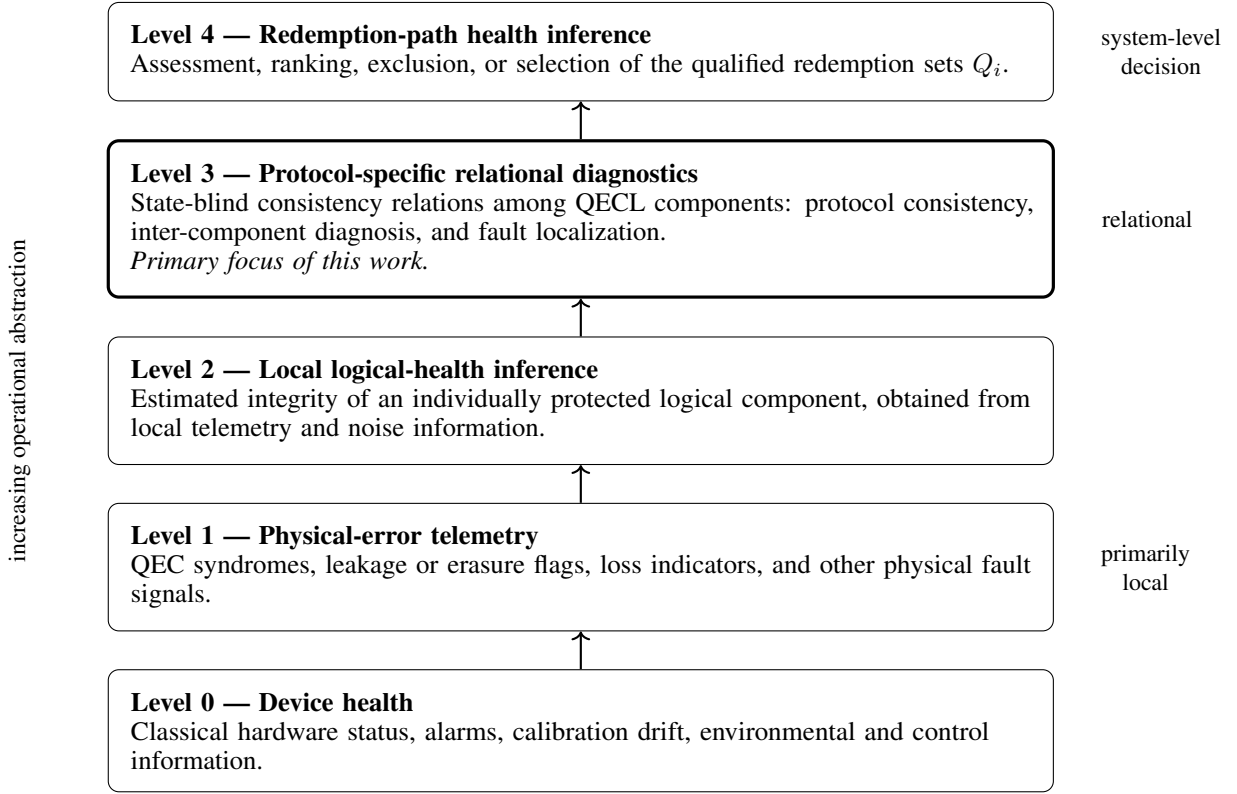
\begin{figure}
\centering
\begin{tikzpicture}[
    node distance=5mm,
    levelbox/.style={
        draw,
        rounded corners,
        align=left,
        text width=0.72\linewidth,
        minimum height=11mm,
        inner sep=3mm
    },
    focusbox/.style={
        levelbox,
        very thick
    },
    arrow/.style={
        ->,
        thick
    }
]

\node[levelbox] (L0) {
\textbf{Level 0 --- Device health}\\
Classical hardware status, alarms, calibration drift, environmental and control information.
};

\node[levelbox, above=of L0] (L1) {
\textbf{Level 1 --- Physical-error telemetry}\\
QEC syndromes, leakage or erasure flags, loss indicators, and other physical fault signals.
};

\node[levelbox, above=of L1] (L2) {
\textbf{Level 2 --- Local logical-health inference}\\
Estimated integrity of an individually protected logical component, obtained from local telemetry and noise information.
};

\node[focusbox, above=of L2] (L3) {
\textbf{Level 3 --- Protocol-specific relational diagnostics}\\
State-blind consistency relations among QECL components:
protocol consistency, inter-component diagnosis, and fault localization.\\
\textit{Primary focus of this work.}
};

\node[levelbox, above=of L3] (L4) {
\textbf{Level 4 --- Redemption-path health inference}\\
Assessment, ranking, exclusion, or selection of the qualified redemption sets \(Q_i\).
};

\draw[arrow] (L0.north) -- (L1.south);
\draw[arrow] (L1.north) -- (L2.south);
\draw[arrow] (L2.north) -- (L3.south);
\draw[arrow] (L3.north) -- (L4.south);

\node[rotate=90, anchor=south] at ([xshift=-9mm]L2.west) {
\small increasing operational abstraction
};

\node[align=center, anchor=west] at ([xshift=5mm]L1.east) {
\small primarily\\
\small local
};

\node[align=center, anchor=west] at ([xshift=5mm]L3.east) {
\small relational
};

\node[align=center, anchor=west] at ([xshift=5mm]L4.east) {
\small system-level\\
\small decision
};

\end{tikzpicture}
\caption{Hierarchy of diagnostic information in distributed QECL. Levels~0--2 summarize local device and logical-health information, Level~3 denotes protocol-specific relational diagnostics, and Level~4 translates the available evidence into redemption-path decisions.}
\label{fig:diagnostic-hierarchy}
\end{figure}


\subsection{A Hierarchy of Diagnostic Information}
\label{subsec:diagnostic-hierarchy}

The distributed setting gives rise to diagnostic information at several levels (Figure~\ref{fig:diagnostic-hierarchy}). Levels~0--2 are primarily local, Level~3 is induced by the QECL protocol itself, and Level~4 turns the available evidence into a system-level redemption decision.

\begin{description}
\item[Level 0: Device health.] Classical information about the hardware supporting a quantum component, such as alarms or memory-device status.

\item[Level 1: Physical-error telemetry.] Quantum-specific evidence about faults affecting a stored or transmitted resource, including signals from the underlying protection layer.

\item[Level 2: Local logical-health inference.] An estimate of the integrity of an individually protected logical component, obtained from local telemetry, syndrome information, and an assumed noise model.

\item[Level 3: Protocol-specific \textit{relational} diagnostics.] State-independent relations imposed by the QECL encoding can be interrogated without resolving the protected logical state. We distinguish three increasing levels of resolution:
\begin{itemize}
\item \textit{protocol consistency}: a prescribed relation is tested, and a violation implicates its support without attributing the fault;
\item \textit{inter-component diagnosis}: overlapping relations confine the fault to a subset smaller than the support of an individual check;
\item \textit{fault localization}: the syndrome identifies a component or, when full separation is impossible, a well-defined diagnostic equivalence class, possibly together with an error type.
\end{itemize}
The progression is one of resolution rather than intrusiveness: all three use state-independent relations and leave the protected logical state unresolved.

\item[Level 4: Redemption-path health inference.] The available evidence is translated into statements about the viability of the qualified sets in \(\Gamma_{\mathrm{red}}\). If \(\mathcal I\) denotes the information collected from Levels~0--3, one may consider
\[
p_i^{\mathrm{red}} = \Pr\!\left(Q_i\ \text{successfully redeems the protected state}\mid\mathcal I\right),
\]
and rank, exclude, or select redemption paths accordingly.
\end{description}

The present work focuses on Level~3. Level~4 provides its operational criterion: relational information is useful insofar as it contributes to a safe decision about one or more future redemption paths.

\section{A Prototypical Duplicated-Syndrome QECL Scheme}
\label{sec:prototype}

Before introducing the general formalism, we illustrate the mechanism on a deliberately simple QECL construction modelled on the three-qubit repetition code and the repetition blocks underlying the Shor code~\cite{shor1995scheme}. The signals carry a coherently replicated binary value, and the checks comparing them are the familiar parity generators \(Z_{S_i}Z_{S_j}\). Unlike an ordinary repetition code, however, these signals are encrypted components rather than readable copies and become usable only in combination with a key.

The example isolates how redundancy introduced for alternative redemption paths can simultaneously generate deterministic relations among encrypted components. It also makes transparent the distinction between fault detection, localization, redemption-irrelevant errors, and common-mode faults that preserve the tested relations.

\subsection{Encoding and Redemption Structure}
\label{subsec:prototype-encoding}

Let the unknown input qubit be
\(
|\psi\rangle_A=\alpha|0\rangle_A+\beta|1\rangle_A 
\),
where \(A\) will subsequently play the role of a common quantum key. Consider first two signal qubits \(S_1\) and \(S_2\). The duplicated-syndrome encoding produces the joint state
\[
|\Psi^{(2)}\rangle
=
\frac{1}{\sqrt{2}}
\left(
|\psi\rangle_A |00\rangle_{S_1S_2}
+
Z_A|\psi\rangle_A |11\rangle_{S_1S_2}
\right).
\]
The two signal qubits therefore carry the same coherently encoded binary syndrome: both are \(0\) in the branch containing \(|\psi\rangle_A\) and both are \(1\) in the branch containing \(Z_A|\psi\rangle_A\).

Each pair
\[
Q_i=\{A,S_i\}, \qquad i\in\{1,2\},
\]
constitutes a redemption set. Recovery through \(Q_i\) is obtained by applying a controlled-\(Z\) operation from \(S_i\) to \(A\). Indeed,
\[
\operatorname{CZ}_{S_i\rightarrow A}
|\Psi^{(2)}\rangle
=
|\psi\rangle_A
\otimes
\frac{|00\rangle_{S_1S_2}+|11\rangle_{S_1S_2}}{\sqrt{2}}
\]
up to the irrelevant disposition of the signal subsystem after redemption. Thus either signal can be combined with the common key \(A\) to recover the same unknown state.
In the terminology of Subsection~\ref{subsec:redemption-paths}, the redemption family is therefore
\[
\Gamma_{\mathrm{red}}
=
\left\{
Q_1,Q_2
\right\},
\qquad
Q_i=\{A,S_i\}.
\]
The qubit \(A\) is shared by all redemption paths and plays the role of the common key, whereas \(S_1\) and \(S_2\) are path-specific encrypted clones. A failure affecting \(S_i\) can therefore invalidate \(Q_i\) while leaving the alternative path potentially usable; a failure affecting \(A\), by contrast, threatens both redemption opportunities.

The same construction extends directly to \(m\) encrypted signals,
\[
|\Psi^{(m)}\rangle
=
\frac{1}{\sqrt{2}}
\left(
|\psi\rangle_A |0\rangle^{\otimes m}
+
Z_A|\psi\rangle_A |1\rangle^{\otimes m}
\right),
\]
with redemption sets
\[
Q_i=\{A,S_i\},
\qquad
i=1,\ldots,m.
\]
The important structural feature is not merely the multiplicity of the sets \(Q_i\), but \textit{the coherent repetition of the same syndrome} across all signal qubits.
Equivalently, the encoding can be viewed as a coherent association between
two operator branches acting on the source qubit and two repetition states
of the signal register:
\begin{eqnarray*}
I_A|\psi\rangle_A
&\longleftrightarrow&
|0\rangle_{S_1}\otimes\cdots\otimes|0\rangle_{S_m},
\\
Z_A|\psi\rangle_A
&\longleftrightarrow&
|1\rangle_{S_1}\otimes\cdots\otimes|1\rangle_{S_m}.
\end{eqnarray*}
Thus the binary branch label associated with the action of \(I_A\) or
\(Z_A\) is coherently repeated across all signal qubits. As a consequence,
every valid encoded state satisfies the pairwise relations
\[
Z_{S_i}Z_{S_j}=+1,
\qquad i\neq j,
\]
\textit{independently of the unknown state} \(|\psi\rangle_A\).

These state-independent relations are not required for redemption itself; they arise from the multiplicity of alternative redemption components and can be used as diagnostic consistency checks.

\subsection{Two-Clone Relational Diagnosis vs.\ Three-Clone Single-Fault Localization}
\label{subsec:clone-detection-localization}

For two encrypted clones, the duplicated syndrome provides the deterministic relation
\[
G_{12}=Z_{S_1}Z_{S_2}, \qquad G_{12}|\Psi^{(2)}\rangle=|\Psi^{(2)}\rangle .
\]
A single \(X\) or \(Y\) error on either signal anticommutes with \(G_{12}\) and therefore changes the measured parity from \(+1\) to \(-1\), whereas a \(Z\) error leaves the parity unchanged. Thus
\[
X_{S_i},Y_{S_i}\Longrightarrow G_{12}=-1,
\qquad
Z_{S_i}\Longrightarrow G_{12}=+1 .
\]
The distinction is already redemption oriented. An \(X\) or \(Y\) component on the signal can propagate through the decoder and corrupt the recovered state, whereas a signal \(Z\) error does not affect the state redeemed on \(A\). The parity check therefore detects the relevant single-signal Pauli classes without inspecting the unknown logical state.

With only two clones, however, the outcome \(G_{12}=-1\) establishes an inconsistency but does not identify whether \(S_1\) or \(S_2\) is responsible. This changes with three encrypted clones. For
\[
|\Psi^{(3)}\rangle
=
\frac{1}{\sqrt{2}}
\left(
|\psi\rangle_A|000\rangle_{S_1S_2S_3}
+
Z_A|\psi\rangle_A|111\rangle_{S_1S_2S_3}
\right),
\]
consider the two independent checks
\[
G_{12}=Z_{S_1}Z_{S_2},
\qquad
G_{23}=Z_{S_2}Z_{S_3}.
\]
Under the single-\(X/Y\)-fault assumption, their outcomes are
\[
\begin{array}{c|cc}
\text{fault location} & G_{12} & G_{23} \\
\hline
\text{none} & +1 & +1 \\
S_1 & -1 & +1 \\
S_2 & -1 & -1 \\
S_3 & +1 & -1
\end{array}
\]
so that each signal location produces a distinct syndrome. This is the content of the following elementary statement.

\begin{proposition}[Detection versus localization in the duplicated-syndrome scheme]
\label{prop:prototype-detection}
In the duplicated-syndrome encoding with \(m\) signals, the checks
\(G_{i,i+1}=Z_{S_i}Z_{S_{i+1}}\) are deterministic and state blind, with healthy
outcome \(+1\). Under a single-fault assumption restricted to \(X\)- and
\(Y\)-type signal errors, two signals suffice to \emph{detect} a fault but not
to localize it, whereas three signals produce a distinct syndrome for each
signal location and therefore \emph{localize} it.
\end{proposition}

The syndrome structure is that of the three-qubit repetition code, but its operational use is different: the signals define alternative redemption paths rather than a redundant encoding to be corrected in place. Once an anomalous \(S_i\) is localized, the corresponding path can be excluded whenever another qualified path remains available.

\subsection{From Redemption Redundancy to Diagnostic Redundancy}
\label{subsec:diagnostic-redundancy}

The repetition-like syndrome does not detect every physical deviation, nor should this be required. Diagnostic usefulness must instead be assessed relative to the effect of an error on future redemption. In the present construction, single-signal Pauli faults fall naturally into two classes: \(X\)- and \(Y\)-type errors are detected by the parity checks and can corrupt redemption, whereas \(Z\)-type errors remain undetected but are harmless for the recovered state.

More generally, the parity checks detect differences among the bit-flip components carried by the signals. Writing a Pauli error on the signal register as
\[
E_S=
\bigotimes_{i=1}^{m}
X_{S_i}^{b_i}Z_{S_i}^{c_i},
\qquad
b_i,c_i\in\{0,1\},
\]
the nearest-neighbor check \(G_i=Z_{S_i}Z_{S_{i+1}}\) acquires the syndrome
\(
g_i=(-1)^{b_i\oplus b_{i+1}}
\).

The checks therefore reveal differential \(X/Y\)-type faults, while transformations that act identically on all compared signals preserve the repetition relations. This is the familiar limitation of repetition-code syndromes~\cite{shor1995scheme}.

For the diagnostic purpose considered here, the relevant classes can be summarized as
\[
\begin{array}{c|c|c}
\text{error class} &
\text{relational syndrome} &
\text{redemption consequence} \\
\hline
\text{signal } Z\text{-type only} &
\text{undetected} &
\text{harmless} \\
\text{differential } X/Y\text{-type} &
\text{detected} &
\text{potentially harmful} \\
\text{common-mode } X/Y\text{-type} &
\text{undetected} &
\text{harmful}
\end{array}
\]
The last class represents a genuine ambiguity of the relational checks, although its practical relevance depends on the deployment assumptions. If encrypted clones are stored at distinct locations or placed in sufficiently independent failure domains, simultaneous harmful faults affecting all compared signals can be substantially less likely than isolated faults. Under an independent error model with harmful bit-flip probability \(p_i\) on signal \(S_i\), the probability of the corresponding collective event scales as
\(
\prod_i p_i
\),
and reduces to \(p^m\) in the identically distributed case.

\section{Relational Diagnosis in the Canonical Yamaguchi--Kempf Protocol}
\label{sec:canonical-protocol}

The canonical encrypted-cloning protocol of Yamaguchi and Kempf~\cite{yamaguchi2026encrypted} contains a richer form of coherent redundancy than the prototype. Its diagnostic structure comes from a Bell label that is repeated across signal--key pairs and coherently associated with a Pauli transformation of the source qubit; relations between repeated labels can therefore be state independent without revealing the label itself.

We use \(n=3\) as the main working case because it is the smallest multiplicity with a nontrivial localization syndrome across several redemption paths. Relations that extend to arbitrary \(n\), and the parity dependence of their measurement cost, are treated in Section~\ref{sec:odd-even-diagnostics}.

\subsection{Protocol Structure and Three-Clone Encoding}
\label{subsec:canonical-structure}

The diagnostic properties of the Yamaguchi and Kempf canonical protocol follow from two structural features that should be distinguished from the outset. First, each encrypted signal \(S_i\) is associated with a key component \(K_i\), and the corresponding pair \(S_iK_i\) carries one of four Bell labels. Second, the \emph{same} Bell label is coherently repeated across all such pairs. Redemption, however, is not pairwise: a signal \(S_i\) is redeemed using the complete key block \(K_1K_2K_3\). The resulting asymmetry between path-specific signal qubits and shared key components will become important for the interpretation of diagnostic syndromes.

\subsubsection{Encoding Structure for Three Clones}
\label{subsec:canonical-encoding}

Let
\(
|\psi\rangle_A=\alpha|0\rangle_A+\beta|1\rangle_A
\)
be the unknown input qubit. We denote the Pauli operators by
$\sigma_0=I$,
$\sigma_1=X$,
$\sigma_2=Y$,
$\sigma_3=Z$,
and introduce the Bell basis
\[
|\phi_\mu\rangle
=
(\sigma_\mu\otimes I)|\Phi^+\rangle,
\qquad
|\Phi^+\rangle
=
\frac{|00\rangle+|11\rangle}{\sqrt{2}}.
\]
Up to convention-dependent phase factors \(\alpha_\mu\), the 3-clone encoded state can be written as
\[
|\Psi_{\mathrm{enc}}^{(3)}\rangle
=
\frac{1}{2}
\sum_{\mu=0}^{3}
\alpha_\mu^{-1}
\sigma_\mu^{(A)}|\psi\rangle_A
\otimes
|\phi_\mu\rangle_{S_1K_1}
\otimes
|\phi_\mu\rangle_{S_2K_2}
\otimes
|\phi_\mu\rangle_{S_3K_3}.
\]
Following the original phase convention~\cite{yamaguchi2026encrypted}, we take
\[
\alpha_0=1,
\qquad
\alpha_1=\alpha_3=i,
\qquad
\alpha_2=-i^{\,n+1}.
\]
For the three-clone case considered explicitly below, this gives
\(\alpha_0=1\), \(\alpha_1=i\), \(\alpha_2=-1\), \(\alpha_3=i\).
The precise phase convention is not essential for the relational checks
considered below, but fixing it will be useful when writing the corresponding
encoded states and observables explicitly.

What matters diagnostically is the coherent association between the Pauli transformation \(\sigma_\mu^{(A)}\) acting on the source register and the identical Bell label \(\mu\) appearing in every signal--key pair.
The diagnostic problem is to identify observables that test relations implied by the common value of \(\mu\) without learning \(\mu\) itself.

\subsubsection{Repeated Bell Labels and Redemption Structure}
\label{subsec:bell-labels-redemption}

Writing the label of pair \(S_iK_i\) as \(\mu_i\), the ideal encoding is supported on the sector
\(
\mu_1=\mu_2=\mu_3
\),
while remaining in coherent superposition over the common value. The repeated object is therefore a two-bit Bell label, rather than the computational-basis syndrome of the prototype.

For three clones, the storage register is
\[
R=S_1K_1S_2K_2S_3K_3.
\]
Let
\[
K=K_1K_2K_3
\]
denote the complete quantum key block. The normal storage-only redemption sets are
\[
Q_i=K\cup\{S_i\},
\qquad
i\in\{1,2,3\}.
\]
Thus each \(S_i\) provides an alternative encrypted clone, whereas the key block \(K\) is shared by all three redemption paths,
\[
\Gamma_{\mathrm{red}}
=
\{Q_1,Q_2,Q_3\}.
\]
A suitable decoding transformation acting on \(Q_i\) restores the unknown state using the selected signal together with the complete key.\footnote{
The family \(\Gamma_{\mathrm{red}}\) is a custody convention rather than a
property of the code. The qualified sets of the underlying access structure are
characterized in~\cite{gianini2026encrypted}: a subset of the storage register
is authorized when it contains both members of at least one signal--key pair
and at least one member of every remaining pair. The two conditions have
different origins. 
\begin{itemize}
\item
A complete pair is needed because the Bell label is stored
in the correlation \emph{within} a pair and neither member carries information
about it on its own. 
\item
A representative of every other pair is needed for a
different reason: 
\begin{itemize}
\item
surrendering a pair in its entirety leaves a perfect record
of the label in the environment, 
\item
which decoheres the branch superposition, and that superposition is where the protected state resides once the transformed
source register has been discarded;
\item
the remaining representatives are therefore
held custodially rather than informatively.
\end{itemize}
\end{itemize}

Consequently, sets outside \(\Gamma_{\mathrm{red}}\) are also qualified: for
three clones, \(\{S_1,K_1,S_2,S_3\}\) and \(\{S_1,K_1,S_2,K_3\}\) both satisfy
the condition. Reserving the complete key block for redemption is a policy
choice, appropriate when the key components are held by a distinct custodian,
and \(\Gamma_{\mathrm{red}}\) should be read accordingly throughout: statements
below about paths lost or preserved are relative to that policy, not to
recoverability in principle.}

This access structure already suggests a diagnostic asymmetry. 
\begin{itemize}
\item
A fault confined to \(S_i\) is primarily associated with the single path \(Q_i\), so localizing that signal can justify excluding one redemption opportunity while retaining the others. 
\item
A fault on a key component \(K_j\), by contrast, affects a subsystem shared by every \(Q_i\) and therefore threatens the complete family of normal redemption
paths, although recovery may remain possible outside
\(\Gamma_{\mathrm{red}}\), through a qualified set that substitutes \(S_j\) for \(K_j\) as the representative of pair \(j\). 
\end{itemize}
Diagnostic localization to a signal--key pair is consequently useful, but distinguishing the signal from the key within that pair can have additional operational value.

\subsubsection{Signal, Key, and Source Registers}
\label{subsec:signal-key-source-registers}

For the subsequent analysis we distinguish the signal register \(S=S_1S_2S_3\), containing the path-specific encrypted clones; the key register \(K=K_1K_2K_3\), shared by the normal storage-only redemption paths; and the transformed source qubit \(A\). Signal-only and key-only observables will first be considered and shown not to provide the required deterministic references. The repeated Bell-label association instead supports inter-pair checks, while access to \(A\) provides a complementary source-assisted sector.

\subsection{Signal-Only and Key-Only Diagnostics}
\label{sec:signal-key-diagnostics}

A natural first attempt is to search for diagnostic relations confined to the
signal register or to the key register, in analogy with the signal-only parity
checks of the prototypical construction. The two encodings, however, place
their redundant labels in structurally different degrees of freedom. 
\begin{itemize}
\item
In the
prototype, the two operator branches \(I_A|\psi\rangle_A\) and
\(Z_A|\psi\rangle_A\) are coherently associated with the repetition states
\(|0\rangle^{\otimes n}\) and \(|1\rangle^{\otimes n}\), so consistency of
the replicated label can be tested directly through signal-only parities. \item
In
the Yamaguchi--Kempf protocol, instead, the four branches
\(\sigma_\mu^{(A)}|\psi\rangle_A\) are coherently associated with repeated Bell states
\(|\phi_\mu\rangle_{S_iK_i}\). The label \(\mu\) is therefore encoded in the association between \(S_i\) and
\(K_i\), rather than in either subsystem separately.\footnote{Within each coherent branch
the pair occupies the Bell state \(|\phi_\mu\rangle\), a joint property that no
marginal records. Indeed \(\rho_{S_i}\), \(\rho_{K_i}\) and, for \(n\geq2\),
\(\rho_{S_iK_i}\) itself are all maximally mixed, so \(\mu\) is not stored in any
reduced state of the stored components; it is defined only relative to the branch
structure shared by every pair.}
\end{itemize}
This already suggests that signal-only and key-only measurements
will not directly test the repeated Bell-label structure: useful
diagnostic observables should instead compare the signal--key association
carried by different pairs.

This expectation is confirmed by the reduced-state structure characterized
in~\cite{gianini2026full}. 

\paragraph{Signal-Only Checks.} Let \(S=S_1\cdots S_n\) denote the complete signal register and let \(y=\langle\psi|Y|\psi\rangle\). Its reduced state is
\[
\rho_S=
\begin{cases}
I^{\otimes n}/2^n, & n\ \mathrm{even},\\[1mm]
\left[I^{\otimes n}+(-1)^{(n-1)/2}yY^{\otimes n}\right]/2^n,
& n\ \mathrm{odd}.
\end{cases}
\]

For even clone multiplicity, every nontrivial signal-only Pauli observable therefore has zero expectation value and no preferred healthy outcome. For odd \(n\), the only nontrivial Pauli component of the complete signal
register is the full-weight\footnote{The
\emph{weight} of a Pauli string is the number of qubits on which it acts
non-trivially, that is, the number of its non-identity tensor factors.
Operationally it is the number of components that must participate in the joint
measurement of the corresponding observable, and it is therefore the natural
cost measure for a relational check.} operator \(Y^{\otimes n}\), whose expectation value depends on the
unknown input through \(y\). Proper signal subsets remain maximally mixed.

The two cases fail diagnostically for complementary reasons. In the even-\(n\) case, signal-only Pauli measurements are state blind but intrinsically random. In the odd-\(n\) case, the only non-random full-register Pauli observable is state dependent and therefore cannot serve as a state-blind diagnostic reference.

\paragraph{Key-Only Checks.} The key register behaves even more uniformly. For \(K=K_1\cdots K_n\),
\[
\rho_K=I^{\otimes n}/2^n,
\]
independently of the parity of \(n\)~\cite{gianini2026full}. Consequently, key-only Pauli measurements also have no deterministic healthy outcome. Moreover, a Pauli fault acting only on the key register leaves this reduced state unchanged, even though such a fault can be operationally important because the key components are shared across the normal redemption paths.

These observations show that the useful diagnostic structure of the canonical protocol is not contained in the signals or keys considered separately.

\subsection{Inter-Pair Signal--Key Relational Diagnostics}
\label{sec:interpair-diagnostics}

Because the Bell label \(\mu\) is carried by the association within each \(S_iK_i\) pair and coherently repeated across pairs, the natural storage-only diagnostic strategy is to compare those associations without determining the common label.

\subsubsection{Bell-Pair Observables}
\label{subsec:bell-pair-observables}

The repeated Bell-label structure of the canonical protocol can be viewed as
a repetition-like encoding over a four-valued alphabet. For each branch
\(\mu\), the transformed source state
\(\sigma_\mu^{(A)}|\psi\rangle_A\) is coherently associated with the same Bell state
\(|\phi_\mu\rangle\) repeated across all signal--key pairs,
\[
\sigma_\mu^{(A)}|\psi\rangle_A
\longleftrightarrow
|\phi_\mu\rangle_{S_1K_1}
\otimes\cdots\otimes
|\phi_\mu\rangle_{S_nK_n}.
\]
The analogy with the binary repetition structure of the prototypical scheme
is useful, but the repeated symbol is now a Bell label rather than a
computational-basis value, and it is encoded relationally within each
signal--key pair.

A Bell label, on the other hand, can be characterized by the eigenvalues of two
commuting Pauli observables. For each pair \(S_iK_i\), we therefore introduce
\[
B_i^X\equiv X_{S_i}X_{K_i},
\qquad
B_i^Z\equiv Z_{S_i}Z_{K_i}.
\]
The four Bell states \(|\phi_\mu\rangle_{S_iK_i}\) are simultaneous
eigenstates of \(B_i^X\) and \(B_i^Z\), so the corresponding pair of
eigenvalues provides two binary coordinates for the four-valued label
\(\mu\), up to the adopted Bell-state convention.

Measuring \(B_i^X\) or \(B_i^Z\) individually, however, would partially
resolve the Bell label carried by pair \(i\). Since the encoded state is a
coherent superposition over \(\mu\), these observables are not suitable as
standalone diagnostic checks. The useful quantities are instead relative
ones, which test whether two pairs carry the same Bell-label coordinates
without determining their common value.

The repeated-label structure nevertheless implies a state-independent relation between different pairs. If pair \(i\) and pair \(j\) carry the same Bell label, then they have identical eigenvalues for both \(B^X\) and \(B^Z\). Their products therefore have fixed eigenvalue \(+1\), independently of the value of \(\mu\).

\subsubsection{State-Blind Inter-Pair Consistency Checks}
\label{subsec:interpair-checks}

For any two pairs \(i\) and \(j\), define the inter-pair observables
\[
G_{ij}^X
\equiv 
B_i^X B_j^X
=
X_{S_i}X_{K_i}X_{S_j}X_{K_j},
\qquad
G_{ij}^Z
\equiv 
B_i^Z B_j^Z
=
Z_{S_i}Z_{K_i}Z_{S_j}Z_{K_j}.
\]
These are deterministic state-blind observables, for every clone multiplicity.

\begin{proposition}[Inter-pair consistency checks]
\label{prop:interpair-deterministic}
For every \(n\), every \(i\neq j\) and every input state \(|\psi\rangle_A\),
\[
G_{ij}^X|\Psi_{\mathrm{enc}}^{(n)}\rangle
=
|\Psi_{\mathrm{enc}}^{(n)}\rangle,
\qquad
G_{ij}^Z|\Psi_{\mathrm{enc}}^{(n)}\rangle
=
|\Psi_{\mathrm{enc}}^{(n)}\rangle .
\]
\end{proposition}

\begin{proof}
The Bell state \(|\phi_\mu\rangle\) is a joint eigenstate of \(B^X=XX\) and
\(B^Z=ZZ\), with eigenvalues \(x_\mu\) and \(z_\mu\) that depend on \(\mu\)
but not on the pair index (Appendix~\ref{app:bell-eigenvalues}). Since the
canonical encoding is supported only on terms carrying the same label \(\mu\)
on every pair, each such term is an eigenstate of \(G_{ij}^X=B_i^XB_j^X\) with
eigenvalue \(x_\mu^2=+1\), and of \(G_{ij}^Z\) with eigenvalue \(z_\mu^2=+1\).
The eigenvalue is thus the same on every branch, so it survives the coherent
superposition and is independent of \(|\psi\rangle_A\).
\end{proof}

These observables therefore test equality of Bell labels without revealing the label itself. The individual quantities \(B_i^X\) and \(B_i^Z\) may vary coherently with \(\mu\), but the relative quantities \(B_i^X B_j^X\) and \(B_i^Z B_j^Z\) do not. In this sense the diagnostic information resides not in the Bell labels themselves, but in their consistency across the redundant pairs.

Only a generating set of pairwise comparisons is required. For three clones we use the adjacent checks
\[
G_{12}^X,\qquad
G_{23}^X,\qquad
G_{12}^Z,\qquad
G_{23}^Z.
\]
The corresponding \(1\)--\(3\) comparisons are redundant, since
\[
G_{13}^X=G_{12}^XG_{23}^X,
\qquad
G_{13}^Z=G_{12}^ZG_{23}^Z.
\]
The healthy syndrome is therefore
\[
(g_{12}^X,g_{23}^X,g_{12}^Z,g_{23}^Z)
=
(+,+,+,+).
\]

\subsubsection{Single-Fault Syndrome Structure}
\label{subsec:interpair-syndrome}

Consider now a single Pauli fault acting on either member of one signal--key pair. An \(X\) error anticommutes with the \(Z\)-type Bell observable of that pair and commutes with its \(X\)-type observable. It therefore flips the outcomes of the \(G^Z\) checks involving the affected pair while leaving the \(G^X\) checks unchanged. Conversely, a \(Z\) error flips the corresponding \(G^X\) checks, while a \(Y\) error flips both families.

For three pairs, the resulting syndromes are shown in Table~\ref{tab:interpair-single-fault}. The table is written at the Bell pair level because the inter-pair checks produce the same syndrome whether the fault acts on the first qubit of the Bell pair\(S_i\) or on the second qubit \(K_i\).

\begin{table}[t]
\centering
\caption{Inter-pair syndrome for a single Pauli fault affecting either qubit of pair \(S_iK_i\). The four columns give the outcomes of \(G_{12}^X\), \(G_{23}^X\), \(G_{12}^Z\), and \(G_{23}^Z\), respectively.}
\label{tab:interpair-single-fault}
\begin{tabular}{c c c c c}
\hline
Fault & \(g_{12}^X\) & \(g_{23}^X\) & \(g_{12}^Z\) & \(g_{23}^Z\) \\
\hline
none        & \(+\) & \(+\) & \(+\) & \(+\) \\
\(X\) on pair 1 & \(+\) & \(+\) & \(-\) & \(+\) \\
\(X\) on pair 2 & \(+\) & \(+\) & \(-\) & \(-\) \\
\(X\) on pair 3 & \(+\) & \(+\) & \(+\) & \(-\) \\
\(Z\) on pair 1 & \(-\) & \(+\) & \(+\) & \(+\) \\
\(Z\) on pair 2 & \(-\) & \(-\) & \(+\) & \(+\) \\
\(Z\) on pair 3 & \(+\) & \(-\) & \(+\) & \(+\) \\
\(Y\) on pair 1 & \(-\) & \(+\) & \(-\) & \(+\) \\
\(Y\) on pair 2 & \(-\) & \(-\) & \(-\) & \(-\) \\
\(Y\) on pair 3 & \(+\) & \(-\) & \(+\) & \(-\) \\
\hline
\end{tabular}
\end{table}

\subsubsection{Pair Localization and Pauli-Error Classification}
\label{subsec:pair-localization}

The four inter-pair checks therefore provide two distinct pieces of information. First, the position of the sign changes localizes the anomalous signal--key pair. Second, the family in which those changes appear identifies the Pauli-error class:
\[
\begin{array}{c|c}
\text{flipped consistency checks} & \text{Pauli class} \\
\hline
G^Z\ \text{only} & X \\
G^X\ \text{only} & Z \\
G^X\ \text{and}\ G^Z & Y 
\end{array}
\]
Thus, under the single-Pauli-fault assumption, the syndrome determines both a pair index \(i\) and an error class \(X\), \(Z\), or \(Y\).

The mechanism extends beyond \(n=3\). For \(n\) signal--key pairs, one may choose a connected set of pairwise comparisons, for example
\[
G_{i,i+1}^X,
\qquad
G_{i,i+1}^Z,
\qquad
i=1,\ldots,n-1.
\]

For general \(n\), a single fault changes the consistency relations incident on the affected pair. More importantly, the inter-pair generators exhaust the diagnostic content of the storage register.

\begin{theorem}[Completeness of the storage-only relational sector]\label{thm:storage-complete}
Let \(V\) be the canonical \(n\)-clone encoding isometry and let
\[
{\mathcal S}_{\mathrm{st}}
=
\left\{
P \ \text{Pauli string on}\ S_1K_1\cdots S_nK_n
\ :\ 
V^{\dagger}PV=\pm I_L
\right\}
\]
be the group of deterministic state-blind observables supported on the storage register alone. Then \({\mathcal S}_{\mathrm{st}}\) has rank \(2(n-1)\) and is generated by the inter-pair checks \(G^{X}_{i,i+1},G^{Z}_{i,i+1}\), \(i=1,\ldots,n-1\).
\end{theorem}

\begin{proof}[Proof sketch]
A Pauli string acts on each pair \(S_iK_i\) as an element of the two-qubit Pauli group. Because the encoded state is supported on the sector in which all Bell labels coincide, and because the four Bell states are the joint eigenstates of \(B^{X}_i\) and \(B^{Z}_i\), a storage-only string has a deterministic expectation value on every encoded state if and only if it acts on each pair as a product \((B^{X}_i)^{x_i}(B^{Z}_i)^{z_i}\) and the resulting eigenvalue is independent of the common label \(\mu\). Since \(B^{X}_i\) and \(B^{Z}_i\) have the same eigenvalues on every pair, the label dependence cancels exactly when
\[
\sum_i x_i\equiv 0,
\qquad
\sum_i z_i\equiv 0
\pmod 2 .
\]
These two parity constraints leave \(2n-2\) free binary parameters, and the inter-pair generators realize a basis of that solution space.
\end{proof}

Two consequences are immediate and independent of the parity of \(n\). First, pair localization together with Pauli classification is optimal among deterministic state-blind measurements confined to the storage register. Second, the within-pair equivalences
\[
X_{S_i}\sim_{\mathrm{diag}} X_{K_i},\qquad
Y_{S_i}\sim_{\mathrm{diag}} Y_{K_i},\qquad
Z_{S_i}\sim_{\mathrm{diag}} Z_{K_i}
\]
are irreducible whenever \(A\) is unavailable. Any finer localization must therefore use information outside the storage-only sector.

\subsubsection{Limits of Within-Pair Localization}
\label{subsec:within-pair-limits}

The inter-pair syndrome therefore cannot distinguish the two members of the affected pair. This is operationally relevant because \(S_i\) is path-specific whereas \(K_i\) belongs to the shared key block; pair localization alone may consequently be insufficient for a redemption decision.
The limitation is not a failure of the Bell-label consistency checks themselves. These observables were deliberately constructed to test whether different pairs carry mutually consistent labels, and they perform that task without resolving the logical Bell label. Their insensitivity to which member of a pair caused the inconsistency reflects the symmetry of the Bell associations they probe.\footnote{Importantly, the within-pair ambiguity does not eliminate the usefulness of the syndrome for
active recovery. Once the \textit{affected pair} and \textit{Pauli class} have been identified,
the corresponding Bell correlation is known to have been altered by a local
Pauli of that class. For an isolated Bell state, the symmetry
\[
(P\otimes I)|\Phi^+\rangle
=
(I\otimes P^T)|\Phi^+\rangle
\]
would make the two error locations locally equivalent. One could therefore
attempt an active correction by applying the identified Pauli class to either
member of the affected pair, even without knowing which qubit was originally
faulty. For a Bell state with a fixed label, this would restore the pair up to
an irrelevant phase. In the QECL encoding, however, the Bell state is a
coherent label of the branch \(\mu\). If the correction is applied to the
wrong member of the pair, the resulting phase can depend on \(\mu\), and hence
become a relative phase between different branches of the encoded
superposition rather than a global phase. Restoring the Bell pair locally is
therefore not, in general, sufficient 
to restore the complete QECL state.}

Access to the source register \(A\) provides the additional information needed to refine this within-pair ambiguity.

\subsection{Diagnostics on Subsets Containing the Input Qubit}
\label{sec:input-diagnostics}

At first sight, using \(A\) for diagnosis may appear problematic. Subsets containing \(A\) are generally informative about the unknown input state, and measuring them without further structure could therefore reveal or disturb the protected quantum information. The relevant question, however, is not whether the subset as a whole is informative, but whether it contains particular observables whose healthy value is independent of the input state. For \(n=3\), the reduced-state structure derived in~\cite{gianini2026full} reveals precisely such observables.

\subsubsection{Informative Subsets Containing the Input}
\label{subsec:informative-input-subsets}

Consider subsets of the form
\[
H=\{A\}\cup C,
\]
where \(C\) contains one representative from each signal--key pair. Let \(q\) denote the number of signal qubits contained in \(C\), so that the remaining \(n-q\) representatives are key qubits.

The informativeness of these subsets depends on the parity of \(n\) and \(q\)~\cite{gianini2026full}. For the three-clone case, subsets with odd \(q\) are fully informative, whereas subsets with even \(q\) are partially informative: their dependence on the unknown state is restricted to the \(y\)-component of its Bloch vector. In particular, the four even-\(q\) subsets
\[
\{A,K_1,K_2,K_3\},
\qquad
\{A,S_1,S_2,K_3\},
\]
\[
\{A,S_1,K_2,S_3\},
\qquad
\{A,K_1,S_2,S_3\}
\]
all belong to this partially informative family.

Their reduced states contain both input-dependent and input-independent Pauli components. This coexistence is diagnostically important: \textit{informativeness is a property of the complete reduced state}, whereas \textit{a diagnostic measurement probes a particular observable within that state}.

\subsubsection{Logical and State-Independent Observable Sectors}
\label{subsec:logical-state-independent-sectors}

For example, the reduced state on \(A K_1K_2K_3\) is~\cite{gianini2026full}
\[
\rho_{A K_1K_2K_3}
=
\frac{1}{16}
\left(
I
+
y\,Z_A X_{K_1}X_{K_2}X_{K_3}
+
Y_A Y_{K_1}Y_{K_2}Y_{K_3}
-
y\,X_A Z_{K_1}Z_{K_2}Z_{K_3}
\right).
\]
The first and third nonidentity terms depend on the unknown Bloch component \(y\). Measuring them would therefore probe the protected state. The middle term is qualitatively different:
\[
\left\langle
Y_A Y_{K_1}Y_{K_2}Y_{K_3}
\right\rangle=+1
\]
\textit{for every input state}.

Thus, although the subset \(A K_1K_2K_3\) is partially informative, it contains a nontrivial observable with a deterministic, state-independent value. The same structure occurs for the other even-\(q\) subsets. This distinction will later be formalized in Section~\ref{sec:relational-diagnosis}: an informative subsystem need not consist exclusively of logical observables; it may also contain an observable sector that acts identically on all encoded logical states.

For the present diagnostic purpose, these input-independent components are precisely the useful ones. They provide healthy reference values without requiring any knowledge of \(|\psi\rangle_A\).

\subsubsection{Source-Containing Transverse Parity Checks}
\label{subsec:source-transverse-checks}

For \(n=3\), the four even-\(q\) subsets identified above yield the following transverse \(Y\)-parity observables:
\[
H_0
=
Y_A Y_{K_1}Y_{K_2}Y_{K_3},
\]
\[
H_{12}
=
Y_A Y_{S_1}Y_{S_2}Y_{K_3},
\qquad
H_{13}
=
Y_A Y_{S_1}Y_{K_2}Y_{S_3},
\]
\[
H_{23}
=
Y_A Y_{K_1}Y_{S_2}Y_{S_3}.
\]
\begin{proposition}[Source-containing transverse checks]
\label{prop:transverse-checks}
For \(\alpha\in\{0,12,13,23\}\) and every input state \(|\psi\rangle_A\),
\[
H_\alpha|\Psi_{\mathrm{enc}}^{(3)}\rangle
=
|\Psi_{\mathrm{enc}}^{(3)}\rangle .
\]
Moreover the four observables are not independent: they satisfy
\(H_0H_{12}H_{13}H_{23}=I\), and \(H_{12}\), \(H_{13}\), \(H_{23}\) are
products of \(H_0\) with inter-pair generators.
\end{proposition}

The deterministic value follows from the reduced states of the four even-\(q\)
subsets; the relations among the four observables are verified in
Appendix~\ref{app:transverse-identities}.

These checks have a useful combinatorial structure. Each contains \(A\) and exactly one representative from every signal--key pair. Across the four checks, the choice between \(S_i\) and \(K_i\) changes in such a way that every storage qubit is associated with a distinct incidence pattern. A fault that anticommutes with \(Y\) therefore flips a characteristic subset of the four outcomes.

Unlike the inter-pair observables of Section~\ref{subsec:interpair-checks}, these checks do not compare Bell labels directly. Rather, they exploit a state-independent sector of the larger source-containing reduced states. They therefore \textit{provide a diagnostically complementary view of the same encoded resource}.

\subsubsection{Single-Fault Localization}
\label{subsec:source-fault-localization}

A single \(X\) or \(Z\) fault on any qubit appearing in one of the \(H\)-checks anticommutes with the local \(Y\) factor and flips the corresponding outcome. The resulting four-bit patterns are shown in Table~\ref{tab:source-containing-syndrome}. Since \(X\) and \(Z\) have the same commutation relation with \(Y\), these checks \textit{do not distinguish the two error classes} but they do \textit{localize the faulty component}.

\begin{table}[t]
\centering
\caption{Syndrome of the source-containing checks for a single \(X\) or \(Z\) fault. The columns correspond to \(H_0\), \(H_{12}\), \(H_{13}\), and \(H_{23}\).}
\label{tab:source-containing-syndrome}
\begin{tabular}{c c c c c}
\hline
Fault location & \(h_0\) & \(h_{12}\) & \(h_{13}\) & \(h_{23}\) \\
\hline
none & \(+\) & \(+\) & \(+\) & \(+\) \\
\(A\)   & \(-\) & \(-\) & \(-\) & \(-\) \\
\(S_1\) & \(+\) & \(-\) & \(-\) & \(+\) \\
\(K_1\) & \(-\) & \(+\) & \(+\) & \(-\) \\
\(S_2\) & \(+\) & \(-\) & \(+\) & \(-\) \\
\(K_2\) & \(-\) & \(+\) & \(-\) & \(+\) \\
\(S_3\) & \(+\) & \(+\) & \(-\) & \(-\) \\
\(K_3\) & \(-\) & \(-\) & \(+\) & \(+\) \\
\hline
\end{tabular}
\end{table}

The seven nontrivial locations therefore produce distinct syndromes for \(X\)- or \(Z\)-type faults, resolving the positional ambiguity left by the inter-pair checks. A \(Y\) fault, however, commutes with the local \(Y\) factors and is invisible to this family, while \(X\) and \(Z\) remain indistinguishable.

Although their supporting subsets may be informative about the protected state, the observables \(H_\alpha\) themselves have deterministic state-independent eigenvalues on the healthy encoding. This refinement is available only while \(A\) remains accessible.

\subsection{Combined One-Shot Diagnostic Scheme}
\label{sec:combined-diagnostics}

The inter-pair and source-containing families are complementary: the former identify the affected pair and Pauli class, while the latter distinguish \(S_i\) from \(K_i\) for \(X\)- and \(Z\)-type faults. Together they span only five independent deterministic checks, whereas the three-clone encoding admits six. The missing generator completes single-Pauli resolution but has greater measurement weight, motivating comparison with the economical five-generator variant.

\subsubsection{A Complete Generating Set}
\label{subsec:minimal-generating-set}

The four inter-pair checks
\[
G_{12}^X,\qquad G_{23}^X,\qquad G_{12}^Z,\qquad G_{23}^Z
\]
are independent, and by Theorem~\ref{thm:storage-complete} they already generate
the whole storage-only relational sector. Among the four source-containing
checks, only one adds a further generator. Indeed, choosing
\[
H_0=Y_A Y_{K_1}Y_{K_2}Y_{K_3},
\]
the remaining three can be reconstructed as
\[
H_{12}=G_{12}^XG_{12}^ZH_0,
\qquad
H_{23}=G_{23}^XG_{23}^ZH_0,
\qquad
H_{13}
=
G_{12}^XG_{23}^XG_{12}^ZG_{23}^ZH_0,
\]
so that the four source-containing observables satisfy
\(
H_0H_{12}H_{13}H_{23}=I
\).

These five generators do not, however, exhaust the deterministic state-blind
observables of the encoding. The three-clone encoding isometry is Clifford and
maps one logical qubit into the seven physical qubits \(A S_1K_1S_2K_2S_3K_3\),
so its image is a \([[7,1]]\) stabilizer code whose stabilizer group has rank
six. Combined with Theorem~\ref{thm:storage-complete}, which fixes the
storage-only rank at \(2(n-1)=4\), this gives the following count.

\begin{proposition}[Source-assisted sector]
\label{prop:source-assisted-rank}
For the canonical \(n\)-clone encoding, the group of deterministic state-blind
Pauli observables has rank \(2n\). Exactly two of its generators involve the
source register \(A\), for every clone multiplicity \(n\).
\end{proposition}

For \(n=3\) the observable \(H_0\) is one of those two generators; the other is
not contained in the span of the five checks above. A direct search identifies
\(32\) deterministic observables outside that span, of which \(24\) have the
minimum weight five. Each has the same structure: it acts on \(A\), on
\emph{both} members of one distinguished signal--key pair through a mixed
two-qubit Pauli of the form \(Y_SZ_K\) or \(Z_SY_K\), and on one representative
of each remaining pair. A convenient choice is
\[
J_1
=
-\,X_A\,Y_{S_1}Z_{K_1}\,X_{K_2}\,X_{K_3},
\]
whose healthy eigenvalue is \(+1\). No symmetric choice exists at this weight:
which pair plays the distinguished role is arbitrary, and the diagnostic
resolution obtained below does not depend on that choice. The price of this
additional generator is a measurement that couples the source register to both
qubits of one pair, which is heavier than any check considered so far.

The complete diagnostic generating set is therefore
\[
\mathcal G_{\mathrm{diag}}
=
\left\{
G_{12}^X,
G_{23}^X,
G_{12}^Z,
G_{23}^Z,
H_0,
J_1
\right\}.
\]
All six observables commute, and the healthy three-clone encoding lies in their
common \(+1\) eigenspace. We denote the corresponding syndrome by
\[
\mathbf s
=
\left(
g_{12}^X,
g_{23}^X,
g_{12}^Z,
g_{23}^Z,
h_0,
j_1
\right).
\]

Table~\ref{tab:combined-single-pauli-syndrome} gives the syndrome generated by
every single-qubit Pauli fault on the source, signal, or key registers.

\begin{table*}[t]
\centering
\caption{Combined syndrome generated by the complete set
\(\{G_{12}^X,G_{23}^X,G_{12}^Z,G_{23}^Z,H_0,J_1\}\)
under a single-Pauli-fault assumption. All twenty-two rows are distinct: every
single-qubit Pauli fault is uniquely identified. Deleting the last column
returns the economical five-generator variant of
Section~\ref{subsec:diagnostic-equivalence-classes}, whose repeated rows define
its residual diagnostic equivalence classes.}
\label{tab:combined-single-pauli-syndrome}
\begin{tabular}{c c c c c c c}
\hline
Fault &
\(g_{12}^X\) &
\(g_{23}^X\) &
\(g_{12}^Z\) &
\(g_{23}^Z\) &
\(h_0\) &
\(j_1\)\\
\hline
none              & \(+\) & \(+\) & \(+\) & \(+\) & \(+\) & \(+\)\\
\(X_A\)           & \(+\) & \(+\) & \(+\) & \(+\) & \(-\) & \(+\)\\
\(Y_A\)           & \(+\) & \(+\) & \(+\) & \(+\) & \(+\) & \(-\)\\
\(Z_A\)           & \(+\) & \(+\) & \(+\) & \(+\) & \(-\) & \(-\)\\
\hline
\(X_{S_1}\) & \(+\) & \(+\) & \(-\) & \(+\) & \(+\) & \(-\)\\
\(X_{K_1}\) & \(+\) & \(+\) & \(-\) & \(+\) & \(-\) & \(-\)\\
\(Z_{S_1}\) & \(-\) & \(+\) & \(+\) & \(+\) & \(+\) & \(-\)\\
\(Z_{K_1}\) & \(-\) & \(+\) & \(+\) & \(+\) & \(-\) & \(+\)\\
\(Y_{S_1}\) & \(-\) & \(+\) & \(-\) & \(+\) & \(+\) & \(+\)\\
\(Y_{K_1}\) & \(-\) & \(+\) & \(-\) & \(+\) & \(+\) & \(-\)\\
\hline
\(X_{S_2}\) & \(+\) & \(+\) & \(-\) & \(-\) & \(+\) & \(+\)\\
\(X_{K_2}\) & \(+\) & \(+\) & \(-\) & \(-\) & \(-\) & \(+\)\\
\(Z_{S_2}\) & \(-\) & \(-\) & \(+\) & \(+\) & \(+\) & \(+\)\\
\(Z_{K_2}\) & \(-\) & \(-\) & \(+\) & \(+\) & \(-\) & \(-\)\\
\(Y_{S_2}\) & \(-\) & \(-\) & \(-\) & \(-\) & \(+\) & \(+\)\\
\(Y_{K_2}\) & \(-\) & \(-\) & \(-\) & \(-\) & \(+\) & \(-\)\\
\hline
\(X_{S_3}\) & \(+\) & \(+\) & \(+\) & \(-\) & \(+\) & \(+\)\\
\(X_{K_3}\) & \(+\) & \(+\) & \(+\) & \(-\) & \(-\) & \(+\)\\
\(Z_{S_3}\) & \(+\) & \(-\) & \(+\) & \(+\) & \(+\) & \(+\)\\
\(Z_{K_3}\) & \(+\) & \(-\) & \(+\) & \(+\) & \(-\) & \(-\)\\
\(Y_{S_3}\) & \(+\) & \(-\) & \(+\) & \(-\) & \(+\) & \(+\)\\
\(Y_{K_3}\) & \(+\) & \(-\) & \(+\) & \(-\) & \(+\) & \(-\)\\
\hline
\end{tabular}
\end{table*}
%

\subsubsection{Diagnostic Resolution and the Economical Variant}
\label{subsec:diagnostic-equivalence-classes}

Two faults \(E\) and \(E'\) are diagnostically equivalent with respect to a
generating set when they produce the same syndrome,
\[
E\sim_{\mathrm{diag}}E'
\quad\Longleftrightarrow\quad
\mathbf s(E)=\mathbf s(E').
\]
The resolution of a diagnostic scheme is the partition of the candidate fault
set induced by this relation. For the complete set, that partition is trivial.

\begin{proposition}[Complete single-fault resolution]
\label{prop:complete-resolution}
Let \(\mathcal F\) be the set consisting of the identity and of all single-qubit
Pauli faults on \(A S_1K_1S_2K_2S_3K_3\). The syndrome map induced by
\(\mathcal G_{\mathrm{diag}}\) is injective on \(\mathcal F\): every element of
\(\mathcal F\) is uniquely identified, and there are no nontrivial diagnostic
equivalence classes.
\end{proposition}

This is read directly from Table~\ref{tab:combined-single-pauli-syndrome}, whose
twenty-two rows are pairwise distinct. In particular the classes
\(\{Y_{S_i},Y_{K_i}\}\), \(\{X_A,Z_A\}\) and \(\{I,Y_A\}\) that arise from the
inter-pair and transverse families alone are not intrinsic limitations of
relational diagnosis in this protocol: they are consequences of measuring only
five of the six available generators.

The six generators are moreover not redundant: no smaller set of deterministic
state-blind checks achieves this resolution. To state this precisely it is
convenient to record the syndrome of a fault as a linear functional on the
diagnostic group. Let \(\mathcal S\) denote the group of deterministic
state-blind Pauli observables of the three-clone encoding, regarded as an
\(\mathbb F_2\)-vector space of dimension six, and for \(E\in\mathcal F\) define
\(\sigma(E)\in\mathcal S^{*}=\mathrm{Hom}(\mathcal S,\mathbb F_2)\) by
\[
\sigma(E)(G)=
\begin{cases}
0, & EG=GE,\\
1, & EG=-GE.
\end{cases}
\]
The map \(\sigma\) is linear in the symplectic representation of the Pauli
group, and a diagnostic set \(\mathcal G\subseteq\mathcal S\) separates \(E\)
from \(E'\) precisely when \(\sigma(E)\) and \(\sigma(E')\) differ somewhere on
\(\mathcal G\).

\begin{proposition}[Minimality of the six-generator set]
\label{prop:minimality}
No subgroup of \(\mathcal S\) of rank five has an injective syndrome map on
\(\mathcal F\). Six deterministic state-blind checks are therefore necessary,
as well as sufficient, for complete single-fault resolution in the three-clone
canonical protocol. The best attainable with five checks is \(19\) distinct
syndromes for the \(22\) elements of \(\mathcal F\).
\end{proposition}

\begin{proof}
A subgroup \(\mathcal H\subseteq\mathcal S\) of rank five is a hyperplane,
hence \(\mathcal H=\ker f\) for some nonzero \(f\in\mathcal S^{*}\), and every
nonzero \(f\) arises from exactly one such subgroup. Two faults \(E,E'\) have
the same \(\mathcal H\)-syndrome exactly when \(\sigma(E)\) and \(\sigma(E')\)
agree on \(\ker f\), that is, when \(\sigma(E)+\sigma(E')\) annihilates
\(\ker f\). The functionals annihilating a hyperplane \(\ker f\) are precisely
\(\{0,f\}\), so, since \(\sigma\) is injective on \(\mathcal F\) by
Proposition~\ref{prop:complete-resolution},
\[
\mathcal H=\ker f
\ \text{identifies some pair of distinct faults}
\iff
f\in\mathcal D
:=
\left\{
\sigma(E)+\sigma(E')
\ :\
E\neq E'\in\mathcal F
\right\}.
\]
A rank-five subgroup with injective syndrome map therefore exists if and only
if \(\mathcal D\) omits at least one of the \(63\) nonzero functionals. Direct
enumeration of the \(\binom{22}{2}=231\) pairs shows that
\(\mathcal D=\mathcal S^{*}\setminus\{0\}\): every nonzero functional is
realized as a difference of two fault syndromes. Hence no rank-five subgroup
separates \(\mathcal F\), and the same enumeration gives \(19\) as the largest
number of distinct syndromes produced by such a subgroup.
\end{proof}

The obstruction is algebraic rather than informational. Five binary checks
provide \(32\) syndromes for \(22\) candidate faults, so the counting bound
\(\lceil\log_2|\mathcal F|\rceil=5\) is not what forbids a five-check solution;
what forbids it is that the available observables form a group, and the
syndromes they induce are consequently constrained.

\paragraph{The economical variant.}
Measuring \(J_1\) requires joint access to the source register and to both
members of a signal--key pair. Where that access is unavailable or too costly,
one may retain the five-generator set
\[
\mathcal G_{\mathrm{diag}}^{\,\mathrm{econ}}
=
\left\{
G_{12}^X,
G_{23}^X,
G_{12}^Z,
G_{23}^Z,
H_0
\right\},
\]
obtained by deleting the last column of
Table~\ref{tab:combined-single-pauli-syndrome}. Its resolution is then
characterized by the nontrivial equivalence classes
\[
\{I,Y_A\},
\qquad
\{X_A,Z_A\},
\qquad
\{Y_{S_i},Y_{K_i}\}\ \ (i=1,2,3),
\]
while all \(X\)- and \(Z\)-type faults on the storage qubits remain individually
resolved. These classes arise for different structural reasons: the transverse
\(Y\)-parity check cannot distinguish \(X_A\) from \(Z_A\); \(Y_A\) commutes with
the whole reduced set and is therefore syndrome-equivalent to the fault-free
case; and within a pair a \(Y\) fault is detected and classified by the
inter-pair checks but is invisible to \(H_0\).

It is worth noting that \(\mathcal G_{\mathrm{diag}}^{\,\mathrm{econ}}\) is not
the five-check set of highest resolution. By
Proposition~\ref{prop:minimality} some rank-five subgroups reach \(19\)
distinct syndromes against the \(17\) produced here; those subgroups, however,
merge classes such as \(\{I,X_{S_2}\}\) and \(\{X_A,X_{K_2}\}\), that is, they
confuse a genuine storage fault with the fault-free case. The classes left by
\(\mathcal G_{\mathrm{diag}}^{\,\mathrm{econ}}\) are of a different nature:
\(Y_A\) acts on a register that is discarded before normal redemption, and
\(\{Y_{S_i},Y_{K_i}\}\) at least localizes the affected pair and its Pauli
class. A five-check scheme with strictly finer partition can therefore be
strictly worse operationally, which is a concrete instance of the criterion
developed in Section~\ref{subsec:redemption-sufficiency}: diagnostic quality is
not measured by the number of faults separated.

\paragraph{Resolution against cost.}
The two variants differ in a way that is easy to quantify. The four inter-pair
generators have weight four and are supported on the storage register alone.
\(H_0\) also has weight four, but requires the source register together with one
representative of each pair. \(J_1\) has weight five and is the only check that
must act on both qubits of a pair simultaneously with \(A\):
\[
\begin{array}{c|c|c}
\text{generator} & \text{weight} & \text{support}\\
\hline
G^{X}_{i,i+1},\,G^{Z}_{i,i+1} & 4 & \text{two pairs, storage only}\\
H_0 & 4 & A + \text{one qubit per pair}\\
J_1 & 5 & A + \text{one full pair} + \text{one qubit per remaining pair}
\end{array}
\]
The additional resolution is therefore not free, and the relevant question is
not whether the finer partition is desirable in the abstract but whether the
distinctions it adds are the ones that matter operationally. That question is
taken up next.

\subsubsection{From Syndrome Information to Redemption Decisions}
\label{subsec:syndrome-redemption-decisions}

The diagnostic distinction that matters most for the normal redemption family \(Q_i=K\cup\{S_i\}\) is whether a fault is path-specific or affects the shared key block. A fault localized to \(S_i\) can be bypassed by excluding \(Q_i\); a fault on \(K_i\) cannot be bypassed merely by choosing another signal, although the larger access structure may permit recovery outside the normal custody convention described in Section~\ref{subsec:canonical-encoding}.

With the economical variant, the distinction is available for \(X\)- and
\(Z\)-type faults but not for \(Y\)-type ones. The residual class
\(\{Y_{S_i},Y_{K_i}\}\) is thus not an incidental ambiguity: it groups together
a fault that can be handled by excluding a single redemption path with one that
threatens all of them. A conservative redemption policy cannot treat pair
localization alone as evidence that the alternative paths are unaffected, and
must either fall back on the shared-key response or acquire further information.

The choice between the two variants is therefore redemption-oriented rather than purely metrological. Measuring \(J_1\) is worthwhile when resolving the signal--key location justifies its additional cost; otherwise a coarser syndrome can suffice whenever every residual class admits a common safe response, as formalized in Section~\ref{subsec:redemption-sufficiency}.

\subsection{Beyond Three Clones: Parity and the Cost of Source-Assisted Checks}
\label{sec:odd-even-diagnostics}

The three-clone instance analyzed above is the smallest setting in which the
diagnostic structure of the canonical protocol can be displayed completely, but
the mechanisms involved are not specific to \(n=3\). Two questions must be kept
apart. The first is which deterministic state-blind relations \emph{exist} for a
given clone multiplicity; the second is what those relations \emph{cost} to
measure. As shown below, the answer to the first question is independent of the
parity of \(n\), whereas the answer to the second is not.

For arbitrary \(n\), a convenient generating family of inter-pair checks is
\[
G_{i,i+1}^{X}
=
B_i^X B_{i+1}^X,
\qquad
G_{i,i+1}^{Z}
=
B_i^Z B_{i+1}^Z,
\qquad
i=1,\ldots,n-1.
\]
These \(2(n-1)\) observables test the consistency of the coherently repeated Bell
label along a connected chain of signal--key pairs. Under a single-Pauli-fault
assumption, the pattern of violated relations localizes the anomalous pair, while
the \(X\)- and \(Z\)-type check families distinguish the Pauli class exactly as in
the three-clone case. By Theorem~\ref{thm:storage-complete} this family is not
merely convenient but complete: it generates the entire storage-only relational
sector, for every \(n\) and for both parities. Equally, the within-pair
equivalences
\[
X_{S_i}\sim_{\mathrm{diag}}X_{K_i},
\qquad
Z_{S_i}\sim_{\mathrm{diag}}Z_{K_i},
\qquad
Y_{S_i}\sim_{\mathrm{diag}}Y_{K_i}
\]
are irreducible whenever the source register \(A\) is unavailable, again for both
parities. Any refinement requires \(A\).

By Proposition~\ref{prop:source-assisted-rank}, exactly two independent
deterministic checks involve \(A\), for every \(n\). Complete single-fault
resolution of the kind established in
Proposition~\ref{prop:complete-resolution} is therefore attainable at every clone
multiplicity, of either parity. What parity governs is the support and weight of
the cheapest generators realizing it.

\paragraph{Odd clone multiplicity.}
Consider a subset
\[
H=\{A\}\cup C,
\]
where \(C\) contains exactly one representative from each signal--key pair, and
let \(q\) be the number of signal qubits in \(C\). For odd \(n\) and even \(q\),
the reduced states characterized in~\cite{gianini2026full} contain, besides terms
depending on the \(y\)-component of the input Bloch vector, a transverse
\(Y\)-Pauli string with a state-independent coefficient. This yields a
deterministic observable of the form
\[
H_C
=
\eta_n
Y_A
\bigotimes_{R_j\in C}Y_{R_j},
\qquad
\eta_n=(-1)^{(n+1)/2},
\]
the sign being fixed by the phase convention of
Section~\ref{subsec:canonical-encoding} so that
\[
H_C|\Psi_{\mathrm{enc}}^{(n)}\rangle
=
|\Psi_{\mathrm{enc}}^{(n)}\rangle.
\]
For \(n=3\) this gives \(\eta_3=+1\), in agreement with the four checks
introduced above; the value alternates with \(n\bmod 4\), so that
\(\eta_5=-1\). There are \(2^{\,n-1}\) such observables, one for each even-\(q\)
subset, and for single-fault diagnosis it is sufficient to retain the all-key
member,
\[
H_0^{(n)}
=
\eta_n
Y_A Y_{K_1}\cdots Y_{K_n},
\]
of weight \(n+1\). Once the inter-pair syndrome has identified an affected pair
\(S_iK_i\) and classified the fault as \(X\) or \(Z\), the outcome of
\(H_0^{(n)}\) distinguishes its location within the pair: an \(X\)- or
\(Z\)-type fault on \(K_i\) anticommutes with \(H_0^{(n)}\), whereas the
corresponding fault on \(S_i\) does not. The set
\[
\left\{
G_{i,i+1}^{X},
G_{i,i+1}^{Z}
\right\}_{i=1}^{n-1}
\cup
\left\{
H_0^{(n)}
\right\}
\]
is thus the general-\(n\) form of the economical variant of
Section~\ref{subsec:diagnostic-equivalence-classes}: it resolves single \(X\)-
and \(Z\)-type storage faults completely, while \(Y_{S_i}\) and \(Y_{K_i}\)
remain equivalent because both commute with the transverse \(Y\)-check. The
second source-assisted generator, which removes that last ambiguity, has weight
\(n+2\); for odd \(n\) it may be taken in the form
\[
J_1^{(n)}
=
-\,X_A\,Y_{S_1}Z_{K_1}\,X_{K_2}\cdots X_{K_n},
\]
verified to be deterministic and independent of the preceding generators.

\paragraph{Even clone multiplicity.}
The first even case, \(n=4\), shows what changes. The storage-only signal
register is now maximally mixed,
\[
\rho_{S_1S_2S_3S_4}
=
\frac{I^{\otimes4}}{16},
\]
so the \(y\)-channel leakage present for odd clone multiplicity disappears. At
the same time, the corresponding source-containing subsets no longer exhibit the
state-independent transverse \(Y\)-sector. For example, for
\[
H=\{A,K_1,K_2,K_3,K_4\},
\]
one obtains~\cite{gianini2026full}
\[
\rho_H
=
\frac{1}{32}
\left(
I
-
z\,Y_A X_{K_1}X_{K_2}X_{K_3}X_{K_4}
+
x\,Z_A Y_{K_1}Y_{K_2}Y_{K_3}Y_{K_4}
-
y\,X_A Z_{K_1}Z_{K_2}Z_{K_3}Z_{K_4}
\right),
\]
in which every nonidentity Pauli component is weighted by an unknown Bloch
coordinate. An exhaustive search over all Pauli strings supported on \(\{A\}\cup
C\), with \(C\) containing one representative per pair, confirms that for even
\(n\) no deterministic observable of that support exists: there is no
counterpart of \(H_0^{(n)}\) at weight \(n+1\).

This does \emph{not} mean that the source-assisted sector is absent. By
Proposition~\ref{prop:source-assisted-rank} two independent source-assisted
generators exist for \(n=4\) as well, and an exhaustive search locates them: they
have minimum weight \(n+2=6\) and each acts on \(A\), on both members of one
signal--key pair, and on one representative of each remaining pair, exactly as
\(J_1\) does for \(n=3\). Adding them to the inter-pair family resolves every
single-qubit Pauli fault uniquely, so complete resolution is available for
\(n=4\) no less than for \(n=3\).

\paragraph{What parity actually controls.}
The comparison can now be stated precisely. For every \(n\), the storage-only
sector has rank \(2(n-1)\) and two further generators involve \(A\). For odd
\(n\), one of the two can be chosen of weight \(n+1\), acting on a single qubit
per pair; for even \(n\) no such choice exists and both have minimum weight
\(n+2\), necessarily acting on both members of some pair. The second generator
has weight \(n+2\) in either case:
\[
\begin{array}{c|c|c}
 & \text{first source-assisted generator} & \text{second}\\
\hline
n\ \text{odd}  & \text{weight }n+1,\ \text{one qubit per pair} & \text{weight }n+2\\
n\ \text{even} & \text{weight }n+2,\ \text{one full pair} & \text{weight }n+2
\end{array}
\]
\textit{Parity therefore governs the cost of relational diagnosis, not its
attainable resolution}. Odd clone multiplicity offers a cheap partial refinement
that even multiplicity does not; \textit{the complete refinement costs the same in both
cases}. The three-clone results should accordingly be read as the smallest
instance of a mechanism available at every multiplicity. The implications for the relation between confidentiality and
relational diagnosability are taken up in
Section~\ref{subsec:secrecy-diagnosability}.

The overall structure is summarized in Figure~\ref{fig:yk-relational-diagnostics}.


\begin{figure*}
\centering
\resizebox{\textwidth}{!}{%
\begin{tikzpicture}[
    x=1cm,
    y=1cm,
    font=\small,
    >=Latex,
    pair/.style={
        draw,
        rounded corners=2pt,
        minimum width=2.0cm,
        minimum height=1.15cm,
        align=center,
        inner sep=5pt
    },
    source/.style={
        draw,
        rounded corners=2pt,
        minimum width=2.0cm,
        minimum height=0.95cm,
        align=center,
        inner sep=5pt
    },
    check/.style={
        draw,
        rounded corners=2pt,
        minimum width=3.05cm,
        minimum height=0.90cm,
        align=center,
        inner sep=5pt
    },
    result/.style={
        draw,
        rounded corners=2pt,
        minimum width=3.30cm,
        minimum height=0.90cm,
        align=center,
        inner sep=5pt
    },
    note/.style={
        align=center,
        font=\scriptsize
    },
    paneltitle/.style={
        font=\bfseries,
        anchor=west
    },
    arrow/.style={
        ->,
        thick
    },
    rel/.style={
        <->,
        thick
    },
    dashedbox/.style={
        draw,
        dashed,
        rounded corners=2pt,
        minimum width=5.0cm,
        minimum height=1.0cm,
        align=center,
        inner sep=6pt
    }
]


\node[paneltitle] at (-8.2,5.75)
{(a) Coherently repeated Bell-label structure};

\node[source] (A) at (-6.55,4.35)
{\(\sigma_\mu^{(A)}|\psi\rangle_A\)};

\node[pair] (P1) at (-3.15,4.35)
{\(S_1\qquad K_1\)\\[1mm]
 \(|\phi_\mu\rangle_{S_1K_1}\)};

\node[pair] (P2) at (0.05,4.35)
{\(S_2\qquad K_2\)\\[1mm]
 \(|\phi_\mu\rangle_{S_2K_2}\)};

\node at (2.55,4.35) {\(\cdots\)};

\node[pair] (Pn) at (5.55,4.35)
{\(S_n\qquad K_n\)\\[1mm]
 \(|\phi_\mu\rangle_{S_nK_n}\)};

\draw[arrow]
(A.east) --
node[above, font=\scriptsize, align=center]
{coherent\\association}
(P1.west);

\draw[rel]
(P1.east) --
node[above, font=\scriptsize]
{same \(\mu\)}
(P2.west);

\draw[rel]
(P2.east) -- (1.65,4.35);

\draw[rel]
(3.45,4.35) --
node[above, font=\scriptsize]
{same \(\mu\)}
(Pn.west);

\draw[
    decorate,
    decoration={brace,mirror,amplitude=5pt},
    thick
]
(-4.45,3.48) -- (6.85,3.48)
node[midway,below=7pt,note]
{repetition-like structure over a four-valued Bell alphabet};

\node[note, align=center] at (0.8,2.72)
{the Bell label is carried by the correlation inside each \(S_iK_i\) pair,
not by \(S_i\) or \(K_i\) separately};


\draw[densely dotted] (-8.2,2.20) -- (8.2,2.20);


\node[paneltitle] at (-8.2,1.65)
{(b) Parity-independent inter-pair diagnosis};

\node[pair] (Q1) at (-6.65,0.20)
{\(S_1\qquad K_1\)\\
 \(B_1^X,\ B_1^Z\)};

\node[pair] (Q2) at (-3.45,0.20)
{\(S_2\qquad K_2\)\\
 \(B_2^X,\ B_2^Z\)};

\node at (-0.95,0.20) {\(\cdots\)};

\node[pair] (Qn) at (1.75,0.20)
{\(S_n\qquad K_n\)\\
 \(B_n^X,\ B_n^Z\)};

\draw[rel]
(Q1.east) --
node[above, font=\scriptsize, align=center]
{\(G_{12}^{X}\)\\[1mm]\(G_{12}^{Z}\)}
(Q2.west);

\draw[rel]
(Q2.east) -- (-1.50,0.20);

\draw[rel]
(-0.35,0.20) --
node[above, font=\scriptsize, align=center]
{\(G_{n-1,n}^{X}\)\\[1mm]\(G_{n-1,n}^{Z}\)}
(Qn.west);

\node[check] (arb) at (-2.45,-1.20)
{\(2(n-1)\) chain generators\\
available for arbitrary \(n\)};

\node[result] (pairloc) at (5.15,0.55)
{anomalous pair \(i\)\\
+\;Pauli class \(X/Z/Y\)};

\node[result] (ambiguity) at (5.15,-0.75)
{residual ambiguity\\
\(S_i\;?\;K_i\)};

\draw[arrow] (Qn.east) -- (pairloc.west);
\draw[arrow] (pairloc.south) -- (ambiguity.north);

\node[note] at (1.15,-1.98)
{compare Bell-label coordinates without resolving the common value of \(\mu\)};


\draw[densely dotted] (-8.2,-2.15) -- (8.2,-2.15);


\node[paneltitle] at (-8.2,-2.70)
{(c) Source-assisted refinement: same resolution, parity-dependent cost};


\node[font=\bfseries] at (-4.25,-3.48)
{Odd \(n\)};

\node[check] (Hodd) at (-4.25,-4.48)
{\(
H_0^{(n)}
=
\eta_n Y_A Y_{K_1}\cdots Y_{K_n}
\)\\
deterministic healthy value};

\node[note] at (-4.25,-5.29)
{state-independent transverse \(Y\)-sector};

\node[result] (oddinput) at (-6.35,-6.63)
{inter-pair syndrome:\\
pair \(i\) + \(X/Z\) class};

\node[result] (Sodd) at (-2.0,-6.05)
{\(S_i\) fault\\
\(H_0^{(n)}\) unchanged};

\node[result] (Kodd) at (-2.00,-7.20)
{\(K_i\) fault\\
\(H_0^{(n)}\) flipped};

\draw[arrow] (oddinput.east) -- ++(0.55,0) |- (Sodd.west);
\draw[arrow] (oddinput.east) -- ++(0.55,0) |- (Kodd.west);

\node[note, align=center] at (-4.60,-8.25)
{\(X_{S_i}\not\sim_{\mathrm{diag}}X_{K_i}\), \quad
 \(Z_{S_i}\not\sim_{\mathrm{diag}}Z_{K_i}\)\\
 \(Y_{S_i}\sim_{\mathrm{diag}}Y_{K_i}\) until the\\
 weight-\((n{+}2)\) generator is added};


\draw[densely dashed] (0,-3.15) -- (0,-8.35);

\node[font=\bfseries] at (4.15,-3.48)
{Even \(n\)};

\node[dashedbox, text width=4.5cm] (Heven) at (4.85,-4.48)
{no weight-\((n{+}1)\) reference:\\
cheapest source-assisted\\
check has weight \(n{+}2\)};

\node[note, align=center] at (4.15,-5.55)
{two source-assisted generators still exist};

\node[result] (eveninput) at (2.10,-6.38)
{inter-pair syndrome:\\
pair \(i\) + Pauli class};

\node[result] (evenamb) at (6.05,-6.38)
{full within-pair\\resolution, at weight \(n{+}2\)};

\draw[arrow] (eveninput.east) -- (evenamb.west);

\node[note, align=center] at (4.15,-7.65)
{the within-pair equivalences persist only if\\
the source-assisted generators are not measured};


\node[
    draw,
    rounded corners=2pt,
    minimum width=15.6cm,
    minimum height=1.05cm,
    align=center,
    inner sep=6pt
] (synthesis) at (0,-9.35)
{\textbf{Diagnostic resolution}
\(\;\longrightarrow\;\)
diagnostic equivalence classes
\(\;\longrightarrow\;\)
redemption decision};


\end{tikzpicture}%
}

\caption{Relational diagnostic structure of the canonical Yamaguchi--Kempf protocol. The coherently repeated Bell label yields \(2(n-1)\) storage-only inter-pair generators that localize an anomalous pair and identify its Pauli class while leaving a signal--key ambiguity. Two additional source-assisted generators remove that ambiguity at every clone multiplicity. Clone parity affects their minimum measurement weight, not the attainable single-fault resolution.}
\label{fig:yk-relational-diagnostics}
\end{figure*}

\subsection{Path Exclusion and Active Error Correction}
\label{subsec:path-exclusion-correction}

A diagnosed path-specific signal fault can be handled by excluding its redemption path when an alternative remains; a shared-key fault, or the absence of a safe alternative, can instead make active correction necessary.
The second is \emph{active correction}. This becomes relevant when no alternative path is available or when the diagnosed fault affects a shared key component. For the resolved \(X\)- and \(Z\)-type storage faults in Table~\ref{tab:combined-single-pauli-syndrome}, the syndrome specifies both the qubit and the Pauli class and therefore provides the information required, at the logical level, to apply the inverse Pauli operation. The unresolved \(Y_{S_i}/Y_{K_i}\) class requires more care.\footnote{ Although the symmetry of \textit{an isolated Bell pair} makes equal Pauli actions on its two members locally related, applying a correction to an arbitrarily chosen member of the pair need not restore the complete QECL state: a phase that is irrelevant for a Bell state with fixed label can become a \(\mu\)-dependent relative phase between coherent branches of the encoded superposition. Additional information capable of resolving the signal--key location, or an appropriately designed recovery operation acting on the full encoded structure, may therefore be required.}

Thus the operational purpose of the syndrome is to reduce uncertainty until a safe action --- path selection, correction, or acquisition of further diagnostic information --- becomes available. This criterion is formalized in the next section.

\section{A General Framework for One-Shot Relational Diagnosis}
\label{sec:relational-diagnosis}

The prototype and the canonical protocol motivate four questions that can be separated from their particular implementations: which encoding relations are state blind, which provide deterministic one-shot references, how finely their syndromes distinguish faults, and what resolution is sufficient for a safe redemption decision. We now formalize these questions and use them to define relational diagnosability as a protocol property.

\subsection{Coherent Operator--Label Associations and Relational Constraints}
\label{subsec:operator-label-relations}

A useful starting point is a structural feature shared by the two constructions studied above. In both cases, the encoding establishes a coherent association between transformations of the unknown logical state and labels represented redundantly in other degrees of freedom. A schematic form is
\[
V|\psi\rangle_L
=
\sum_{\lambda\in\Lambda}
c_\lambda
U_\lambda|\psi\rangle_L
\otimes
|\Lambda_\lambda\rangle_R,
\]
where \(V\) is the encoding isometry, \(U_\lambda\) acts on the logical information, and \(|\Lambda_\lambda\rangle_R\) represents the corresponding physical label structure. This expression is not intended as a universal normal form for QECL; it isolates a mechanism that is explicit in both examples considered here.

In the prototype the physical label is a coherently repeated binary value, whereas in the canonical protocol it is a four-valued Bell label represented relationally within each signal--key pair. In both cases the useful observables compare redundant label coordinates without resolving the common label itself; Sections~\ref{sec:prototype} and~\ref{sec:interpair-diagnostics} give the corresponding explicit constructions.

These examples suggest a general diagnostic principle. The information-bearing label itself need not be measured. Instead, one may search for constraints that are satisfied by every admissible label configuration. If an observable \(G\) satisfies
\[
G|\Lambda_\lambda\rangle_R
=
g|\Lambda_\lambda\rangle_R
\]
\textit{with the same} \(g\) \textit{for every} \(\lambda\), then the label value remains unresolved while the relation defining the valid encoding can still be tested.

Literal equality of repeated labels, as in the two examples above, is only one possible source of such constraints. More generally, a QECL encoding may impose parity, stabilizer-like, or other algebraic relations among its components. What matters diagnostically is that the relation be independent of the unknown logical state while remaining sensitive to relevant departures from the encoded structure.

The source-containing checks of Section~\ref{sec:input-diagnostics} also show that not every useful relation needs to act exclusively on a separate label register. The observables \(H_\alpha\) involve the transformed source qubit together with signal and key components, yet their healthy eigenvalues remain independent of the logical state. For this reason, the remainder of the framework is formulated directly in terms of the encoding isometry \(V\), rather than requiring a particular operator--label decomposition.

\subsection{Deterministic v. non-deterministic state-blind observables, informative observables}
\label{subsec:state-blind-pauli}

Let
\(
V:\mathcal H_L\rightarrow\mathcal H_R
\)
be the QECL encoding isometry from the \textit{logical} Hilbert space
\(\mathcal H_L\) to the \textit{physical resource space} \(\mathcal H_R\). Since \(V\) is an isometry,
\(
V^\dagger V=I_L
\).
If \(\rho\) denotes the unknown logical state before encoding, the
corresponding encoded state is
\[
\rho_{\mathrm{enc}}=V\rho V^\dagger.
\]

Consider a Pauli-string observable \(G\) acting on some subset of the
physical resource. Its two outcomes are associated with the following projectors\footnote{Being a tensor product of Pauli factors, \(G\) is Hermitian and
squares to the identity, \(G^2=I\), so its spectrum is contained in
\(\{+1,-1\}\) and a measurement of \(G\) has two outcomes. Writing the spectral
decomposition \(G=\Pi_+-\Pi_-\) together with the completeness relation
\(I=\Pi_++\Pi_-\), and adding or subtracting the two, gives the projectors
above; they are orthogonal and idempotent, again by \(G^2=I\).}
\[
\Pi_\pm=\frac{I\pm G}{2}.
\]
Their probabilities on the encoded state are
\[
p_\pm(\rho)
=
\operatorname{Tr}
\left(
\Pi_\pm\rho_{\mathrm{enc}}
\right)
=
\frac{1}{2}
\left[
1
\pm
\operatorname{Tr}
\left(
\rho V^\dagger G V
\right)
\right].
\]

The operator \(V^\dagger G V\) is the action induced by the physical observable \(G\) on the encoded logical subspace. Its expectation on \(\rho\) equals \(\langle G\rangle_{\rho_{\mathrm{enc}}}\), so it directly determines the measurement probabilities above.

The measurement outcomes are independent of the unknown logical state if
and only if
\[
V^\dagger G V=cI_L
\]
for some real \(c\in[-1,1]\). This motivates the following definition.

\begin{definition}[State-blind and deterministically state-blind observables]
\label{def:state-blind}
A Pauli-string observable \(G\) on the physical resource is \emph{state blind}
with respect to the encoding \(V\) when \(V^\dagger GV=c\,I_L\) for some real
\(c\in[-1,1]\). It is \emph{deterministically state blind}, and provides a \emph{deterministic
healthy reference}, when \(c=+1\) or \(c=-1\).\footnote{Both values are
deterministic; they differ only in which outcome is prescribed. Recall that
\(p_\pm=\tfrac12(1\pm c)\), so \(c=+1\) gives \(p_+=1\), \(p_-=0\) and \(c=-1\)
gives \(p_-=1\), \(p_+=0\): in either case one of the two outcomes occurs with
certainty on every healthy encoded resource, and the other never does. Any
intermediate value leaves both outcomes possible, since
\(\mathrm{Var}(G)=\langle G^{2}\rangle-\langle G\rangle^{2}=\langle
I\rangle-c^{2}=1-c^{2}\), using \(G^{2}=I\), and \(\langle I\rangle = \mathrm{Tr}(I\rho) = \mathrm{Tr}(\rho) = 1\), and this variance vanishes only at the two endpoints. The limiting case \(c=0\) is still state blind, but maximally random, and therefore useless
as a one-shot reference: observing \(-1\) is exactly what a fault-free resource
produces half of the time. Only the deterministic case supplies an outcome
whose violation is, in a single measurement, evidence of a fault.}
\end{definition}
The definition separates three cases relevant here. If \(V^\dagger G V=\pm I_L\), \(G\) supplies a deterministic healthy reference, as do the prototype parities, the Yamaguchi--Kempf inter-pair checks, and the source-containing deterministic checks. If \(V^\dagger G V=cI_L\) with \(|c|<1\), the observable remains state blind but is intrinsically probabilistic; the maximally mixed signal and key marginals give the important case \(c=0\). If \(V^\dagger G V\) is not proportional to the identity, the measurement generally probes the logical state and is informative rather than state blind.

State blindness is therefore a property of an observable, not necessarily of its supporting subsystem. An informative subsystem may contain deterministic state-blind observables, while a maximally mixed subsystem may make local Pauli measurements state blind without supplying any one-shot reference. Subsystem secrecy and diagnostic usefulness are distinct properties.

\subsection{Deterministic One-Shot References and Syndrome Extraction}
\label{subsec:oneshot-references}

The distinction between state blindness and deterministic state blindness is particularly important in the present setting because the diagnosis is \emph{one shot}. The diagnostic system is assumed to receive a single currently stored encoded resource, without a previous measurement record of the same quantum state and without access to repeated identical preparations of the unknown payload. A useful healthy reference must therefore be encoded in the structure of the resource itself.

For this purpose, the strongest case is
\[
V^\dagger G V=gI_L,
\qquad
g\in\{+1,-1\}.
\]
Since \(G\) is Hermitian and unitary, this condition implies
\[
GV=gV.
\]
The complete healthy code space is therefore contained in one eigenspace of \(G\). After absorbing the known sign \(g\) into the definition of the observable, every diagnostic check may be normalized so that
\[
G_aV=V.
\]
A commuting family
\(
\mathcal G=
\{G_1,\ldots,G_r\}
\)
then defines a healthy diagnostic syndrome
\(
\mathbf s_0=(+1,\ldots,+1)
\).

This structure also explains why the corresponding measurements can be performed nondestructively at the logical level. Measuring \(G_a\) applies the projectors
\[
\Pi_\pm^{(a)}
=
\frac{I\pm G_a}{2}.
\]
For every healthy encoded state,
\[
\Pi_+^{(a)}V\rho V^\dagger\Pi_+^{(a)}
=
V\rho V^\dagger,
\qquad
\Pi_-^{(a)}V\rho V^\dagger\Pi_-^{(a)}
=
0.
\]
Thus the measurement does not resolve the unknown logical state and does not further project a valid encoded resource. It reads an eigenspace label that is already fixed by the encoding.

The same property holds for the single-Pauli fault model considered in Sections~\ref{sec:prototype} and~\ref{sec:canonical-protocol}. Let \(E\) be a Pauli fault and define
\[
G_aE
=
\chi_a(E)EG_a,
\qquad
\chi_a(E)\in\{+1,-1\}.
\]
Then
\[
G_aEV
=
\chi_a(E)EV.
\]
The faulty state is therefore itself a definite eigenstate of every diagnostic generator, with syndrome
\[
\mathbf s(E)
=
\left(
\chi_1(E),\ldots,\chi_r(E)
\right).
\]
Syndrome extraction identifies which relational eigenspace the error has moved the resource into, without measuring the logical payload.
State blindness is therefore necessary for the type of non-leaky diagnosis considered here, whereas a deterministic healthy relation provides the one-shot reference that makes the measurement diagnostically actionable.

The condition \(V^{\dagger}GV=\pm I_L\) admits a sharper reading. For a Pauli string \(G\), it holds if and only if \(G\) acts as \(\pm\) the identity on the code space \({\mathcal C}=\mathrm{Im}\,V\), that is, if and only if \(G\) belongs to the stabilizer group \({\mathcal S}({\mathcal C})\) up to sign. Deterministic state-blind Pauli observables are therefore \emph{exactly} the stabilizer elements of the QECL encoding, and a one-shot relational diagnostic set is a chosen subgroup of that stabilizer. When the encoding is Clifford, as is the case for the constructions studied here, \({\mathcal S}({\mathcal C})\) is obtained in polynomial time by propagating the stabilizer generators of the input ancillas through the encoding circuit, and has rank \(N-k\) for an \([[N,k]]\) code.

Thus, \textit{the algebra is that of stabilizer-syndrome extraction in quantum error correction}; \textit{what differs is the operational use}. A conventional QEC syndrome is normally interpreted as information for restoring the encoded block. A QECL relational syndrome may instead be used to determine which redemption opportunity should be avoided or selected, without physically correcting the affected component. The distinctive questions in the QECL setting are consequently not \textit{which relations exist}, but \textit{which subgroup of \(\mathcal S(\mathcal C)\) is accessible} -- given which components remain available and can be measured jointly -- \textit{at what measurement weight}, and \textit{whether the resulting partition suffices} for a safe redemption decision.

\subsection{Diagnostic Coverage, Resolution, and Equivalence Classes}
\label{subsec:diagnostic-resolution}

The language used in this subsection deliberately parallels standard
syndrome-based quantum error correction: a family of commuting checks induces
a syndrome map over a specified fault set, and faults are distinguished
according to their commutation patterns with those checks. The difference lies
in the operational objective: in QECL, the syndrome need not identify a
correction that restores a single encoded block; it may instead provide only
the resolution required to select, exclude, or repair a future redemption
path.

A collection of deterministic checks defines a map from candidate faults to syndromes. Let
\(
\mathcal F
\)
be a fault set containing the identity \(I\), and let
\[
\mathbf s_{\mathcal G}:
\mathcal F
\rightarrow
\{+1,-1\}^{r}
\]
be the syndrome map generated by the commuting family \(\mathcal G\).

The first diagnostic property is \emph{detection}. A fault \(E\) is detected whenever
\[
\mathbf s_{\mathcal G}(E)
\neq
\mathbf s_{\mathcal G}(I).
\]
being \(\mathbf s_{\mathcal G}(I)\), the syndrome of the fault-free case. The corresponding diagnostic blind set is
\[
\mathcal B_{\mathcal G}
=
\left\{
E\in\mathcal F:
\mathbf s_{\mathcal G}(E)
=
\mathbf s_{\mathcal G}(I)
\right\}.
\]
Importantly, membership in \(\mathcal B_{\mathcal G}\) does not by itself imply diagnostic failure. Some undetected faults may be irrelevant to redemption, while others may be harmful. In the prototype, for example, signal \(Z\)-type errors are invisible to the \(Z\)-parity checks but do not compromise the controlled-\(Z\) redemption mechanism. Common-mode transformations can instead preserve the same relations while potentially affecting all redemption paths. The relevant question is therefore not only which faults are invisible, but \textit{what operational effect the invisible faults have}.

Beyond detection, the syndrome determines a partition of the fault set.

\begin{definition}[Diagnostic equivalence and resolution]
\label{def:diagnostic-resolution}
Two faults \(E,E'\in\mathcal F\) are \emph{diagnostically equivalent} with
respect to a diagnostic set \(\mathcal G\) when \(\mathbf s_{\mathcal G}(E)=
\mathbf s_{\mathcal G}(E')\). The induced partition of \(\mathcal F\) into
equivalence classes \([E]_{\mathcal G}\) is the \emph{diagnostic resolution}
of \(\mathcal G\); a diagnostic set is \emph{finer} than another when its
partition refines the other's.
\end{definition}

We write the resulting class as \([E]_{\mathcal G}\). Adding independent checks may refine the partition by separating previously indistinguishable faults.
This viewpoint unifies several levels of diagnosis that appeared in the examples.
\begin{itemize}
\item
\emph{Detection} only requires distinguishing a fault from the healthy class.
\item
\emph{Localization} requires a syndrome class to identify a sufficiently small physical region, such as a particular signal, signal--key pair, site, or redemption path.
\item
\emph{Error classification} requires the syndrome to determine a relevant error type, such as the \(X\), \(Z\), or \(Y\) Pauli class.
\item
\emph{Fault identification} is the strongest case, in which the syndrome class contains only one relevant candidate fault, up to physically irrelevant phases or other equivalences defined by the operational model.
\end{itemize}

The examples above realize these levels in different combinations: the prototype progresses from detection to signal localization, while the canonical storage-only checks identify a signal--key pair and Pauli class before source-assisted checks refine the within-pair location. The syndrome partition records these distinctions more precisely than a generic statement that a scheme ``detects errors.''

\subsection{Redemption-Oriented Diagnostic Sufficiency}
\label{subsec:redemption-sufficiency}

Fault identification is not necessarily the correct endpoint for QECL diagnosis. The purpose of diagnosis is ultimately to support a future redemption decision, and different physical faults may require the same operational response. Conversely, two faults that produce the same syndrome may represent an important ambiguity if they require different actions.

To formalize this distinction, associate with each candidate fault \(E\) a set
\(
\mathcal A(E)
\)
of diagnostic actions that are safe under that fault. An action may, for example, select a particular redemption set, exclude one or more paths, apply a specified correction, or request additional diagnostic information before redemption.

A diagnostic equivalence class
\[
C=[E]_{\mathcal G}
\]
is redemption-resolvable in the sense of the following definition.

\begin{definition}[Redemption-oriented diagnostic sufficiency]
\label{def:redemption-resolvable}
Associate with each candidate fault \(E\) the set \(\mathcal A(E)\) of
diagnostic actions that are safe under \(E\). A diagnostic equivalence class
\(C=[E]_{\mathcal G}\) is \emph{redemption-resolvable} when some action is safe
for every fault compatible with its syndrome, that is, when
\(\bigcap_{E'\in C}\mathcal A(E')\neq\varnothing\).
\end{definition}

Thus the syndrome need not determine which member of \(C\) occurred once it supports a common safe response. This distinguishes physical ambiguity from operational ambiguity. In the prototype, for example, localized \(X_{S_i}\) and \(Y_{S_i}\) faults can share the same action --- exclude \(Q_i\) --- while an undetected \(Z_{S_i}\) fault may safely share the healthy syndrome. By contrast, the canonical class \(\{Y_{S_i},Y_{K_i}\}\) combines a path-specific fault with a shared-key fault and may therefore remain redemption-relevant despite already specifying the affected pair and Pauli class.

Consequently, diagnostic performance is not determined solely by the number of uniquely identified faults: a coarser partition can be sufficient when each class has a common safe action, while a finer one can still be operationally inadequate.

For path-selection problems, the same idea can be expressed directly through the viable redemption sets. Let
\[
\mathcal V(E)
\subseteq
\Gamma_{\mathrm{red}}
\]
denote the redemption sets that remain viable under fault \(E\). A syndrome class \(C\) supports immediate path selection whenever
\[
\bigcap_{E\in C}
\mathcal V(E)
\neq
\varnothing.
\]
Any redemption set in this intersection is safe without determining which fault in \(C\) actually occurred. Active correction can be incorporated by enlarging the action space beyond direct path selection.

The QECL diagnostic objective is therefore to determine enough about the current fault state to identify a safe redemption action; detection, localization, and error classification are intermediate forms of resolution toward that goal.

\subsection{Relational Diagnosability as a Protocol Property}
\label{subsec:relational-diagnosability-framework}

The preceding notions lead naturally to \emph{relational diagnosability}. The term refers to the ability of a QECL encoding to provide state-independent relations among its physical components whose measurement distinguishes relevant departures from the healthy encoding at sufficient resolution for the intended redemption task.

Relational diagnosability is not an absolute property of an encoding in isolation. It is defined relative to at least four ingredients:
\[
\left(
V,
\Gamma_{\mathrm{red}},
\mathcal F,
\mathcal O_{\mathrm{adm}}
\right),
\]
where \(V\) is the \textit{encoding}, \(\Gamma_{\mathrm{red}}\) is the \textit{redemption structure}, \(\mathcal F\) is the \textit{fault model} of interest, and \(\mathcal O_{\mathrm{adm}}\) is the \textit{family of diagnostic observables} that are physically and operationally \textit{admissible}. The latter may depend on which components remain accessible, whether the source register is retained, and which distributed measurements can be implemented.

For a given setting, we call a commuting family
\[
\mathcal G
=
\{G_1,\ldots,G_r\}
\subseteq
\mathcal O_{\mathrm{adm}}
\]
a \textit{one-shot relational diagnostic set} in the sense of the following definition.

\begin{definition}[One-shot relational diagnostic set; relational diagnosability]
\label{def:relational-diagnosability}
Given \((V,\Gamma_{\mathrm{red}},\mathcal F,\mathcal O_{\mathrm{adm}})\), a
commuting family \(\mathcal G=\{G_1,\ldots,G_r\}\subseteq
\mathcal O_{\mathrm{adm}}\) is a \emph{one-shot relational diagnostic set} when
every generator is deterministically state blind,
\[
V^\dagger G_aV=g_aI_L,
\qquad
g_a\in\{+1,-1\},
\]
and its syndrome extraction is compatible with preservation of the encoded
logical information under \(\mathcal F\). The encoding is \emph{relationally
diagnosable} with respect to a specified redemption objective when such a
\(\mathcal G\) exists whose syndrome partition is redemption-resolvable
(Definition~\ref{def:redemption-resolvable}) over \(\mathcal F\).
\end{definition}

The checks therefore need not identify every error uniquely; they must separate
faults until each remaining ambiguity class admits an acceptable common action.

This definition makes relational diagnosability naturally graded rather than purely binary. Several aspects may be used to characterize its strength.

\emph{Diagnostic coverage} describes which redemption-relevant faults leave the healthy syndrome class and which remain relationally invisible.

\emph{Diagnostic resolution} describes the equivalence classes induced by the syndrome and whether they correspond to physical components, pairs, sites, Pauli classes, or individual faults.

\emph{Redemption resolution} asks whether the remaining equivalence classes are sufficiently fine to support path selection, exclusion, correction, or another required operational response.

\emph{Measurement support} characterizes how many and which distributed components must participate in each check. The prototype uses low-weight signal--signal relations, whereas the canonical protocol relies on four-body inter-pair checks and, for odd \(n\), source-containing transverse relations.

\emph{Check complexity} includes the number of independent generators required for a given resolution. In the three-clone canonical protocol, the two initially identified check families span a rank-five subgroup, while complete single-Pauli resolution requires the sixth independent generator.

\emph{Access dependence} records whether diagnostic resolution depends on resources that may not remain available throughout the storage lifecycle. The inter-pair Yamaguchi--Kempf checks require only stored signal and key components, whereas the within-pair refinement of Section~\ref{sec:input-diagnostics} requires access to \(A\).

\emph{Common-mode sensitivity} characterizes faults that preserve all tested relations. Relational checks are naturally sensitive to differential deviations, but collective transformations may lie in the blind class. Whether this is an important limitation depends on the fault model and on the degree of independence among the physical failure domains.

These dimensions make relational diagnosability a property distinct from clone multiplicity, secrecy, and ordinary recovery capability. In particular, the canonical odd--even comparison shows that eliminating a low-weight state-independent relation can leave attainable diagnostic resolution unchanged while increasing its measurement cost. Protocol comparison should therefore include not only access structure and leakage, but also which deterministic relational checks are accessible, at what cost, and whether their residual ambiguities are compatible with the intended redemption policy.

\section{Discussion}
\label{sec:discussion}

The preceding results allow the diagnostic structure of QECL to be compared with secrecy, classical diagnosis, local QEC, and distributed implementation constraints.

\subsection{Encrypted Redundancy as a Diagnostic Resource}
\label{subsec:encrypted-redundancy-resource}

Encrypted redundancy provides alternative redemption opportunities and, when the encoding imposes suitable relations among them, internal diagnostic references. Recovery redundancy and diagnostic redundancy therefore need not scale together: for the canonical protocol, Theorem~\ref{thm:storage-complete} fixes the storage-only relational rank at \(2(n-1)\) and leaves an irreducible within-pair ambiguity.

Behind this lies \textit{a structural feature of pure encodings that has no classical
counterpart}. For any bipartition \(B\mid B^{c}\) of the encoded resource, purity
imposes a two-sided constraint: at most one of the two parts can be authorized,
by no-cloning, and at least one of them must be informative, by the complementarity between recoverability and privacy~\cite{gianini2026full}, since a completely uninformative part forces its
complement to be authorized. Exactly one of three cases therefore occurs,
namely \(B\) authorized, \(B^{c}\) authorized, or both partially informative.
\textit{Neither half of the constraint holds classically}. 
In a Shamir scheme~\cite{shamir1979share} with
threshold \(k<n\), a set of \(k-1\) shares and its complement can both be
uninformative, which purity forbids; 
under plain replication~\cite{patterson1988raid,reich2006lockss} both can be qualified, which no-cloning forbids.

The consequence for diagnosis is that \textit{informativeness is not a property that
can be assigned to components and then accumulated over a subset}. In the
canonical protocol every single-qubit marginal is maximally mixed, and no qubit
holds any part of the protected state; the dependence on the input resides
entirely in the coherences between distinct branch labels. Whether a subset can
redeem the state accordingly depends on its internal correlation structure ---
whether it retains a complete signal--key pair, and whether it denies the
environment a complete record of the label --- and not on any aggregate over
the components it contains. This is not a dynamical statement: \(\rho_B\)
depends on \(B\) alone, and nothing propagates from one part to another. It is
a statement about what determines \(\rho_B\), namely a joint property invisible
at the level of the individual components.

Two features of the analysis follow from this. \textit{Redemption-path health} at
Level~4 of Section~\ref{subsec:diagnostic-hierarchy} \textit{cannot be obtained by
combining independent per-component health estimates}, as it could if the
components behaved classically; and, as will be observed in
Section~\ref{subsec:classical-diagnosis}, the classical lever of enriching the
test graph is unavailable here, the outcomes of relational checks being
algebraically constrained. Non-additivity of health and algebraic dependence of the checks are the same
non-locality seen from the side of the state and from the side of the
observables.  Together they indicate that
\textit{relational diagnosability is not a quantum restatement of classical fault
diagnosis but a subject with its own structure}: the classical theory supplies
the vocabulary and the operational questions, while what determines the answers
--- which relations exist, which of them can be tested jointly, and how their
outcomes constrain one another --- is fixed by the encoding rather than by the
diagnostic procedure laid over it.

\subsection{Secrecy, Diagnosability, and the Cost of Diagnosis}
\label{subsec:secrecy-diagnosability}

Section~\ref{sec:odd-even-diagnostics} shows that clone parity changes the cost of source-assisted diagnosis rather than its attainable single-fault resolution. For odd \(n\), the same parity-dependent Pauli structure that permits a restricted leakage channel also supplies a weight-\((n+1)\) transverse \(Y\)-relation; for even \(n\), both disappear, but two source-assisted generators still exist and complete resolution remains available at weight \(n+2\).

The relevant comparison is therefore between confidentiality and the \emph{cost} of diagnosis, not between confidentiality and diagnosability. The leakage channel and the low-weight diagnostic relation have a common structural origin, but neither causes or bounds the other.

Two consequences follow for protocol comparison, in view of QECL protocol design. 
\begin{itemize}
\item
One should be careful
\textit{not to identify stronger secrecy with universally better protocol behavior, nor weaker secrecy with better diagnosability}: what a leakage-free encoding gives up
here is measurement economy, which may or may not matter depending on whether
distributed joint measurements over a full signal--key pair are available in the
intended custody architecture. 
\item
\textit{Relational diagnosability} is genuinely a
third axis alongside \textit{leakage} and \textit{access structure}, but it must be assessed together with a cost model. A protocol-selection problem may need to consider not only what unauthorized subsets can learn and which subsets can redeem the state, but also which state-blind relations remain measurable, at what weight and support, what fault distinctions they enable, and whether that resolution is sufficient for the intended redemption policy.
\end{itemize}

The general shape of the design problem is therefore clearer than in a pure
trade-off picture, and also harder: since the admissible relations of a Clifford
QECL encoding can be enumerated exhaustively
(Section~\ref{subsec:oneshot-references}), what remains is to select within them a subgroup that is simultaneously \textit{cheap}, \textit{architecturally realizable}, and \textit{redemption-sufficient}. A systematic optimization of that selection lies beyond the scope of the present work.

\subsection{Relation to Classical Fault Diagnosis and Error Location}
\label{subsec:classical-diagnosis}

Diagnosability has a long history as a design criterion outside quantum
information, and two classical traditions are close enough to the present
setting to be worth making explicit.

\paragraph{System-level diagnosis.}
The theory of self-diagnosing multiprocessor systems, initiated by Preparata,
Metze and Chien~\cite{preparata1967connection}, models a system as a graph
whose vertices are units and whose edges are tests, and asks which test
assignments allow a faulty set to be identified. Its comparison-model
variant~\cite{sengupta1992comparison} is the closest analogue of the mechanism
studied here: a comparator receives the outputs of two units and reports
agreement or disagreement, without evaluating either output on its own. The
inter-pair observables of Section~\ref{subsec:interpair-checks} are comparators
in exactly this sense. They test whether the Bell labels of two signal--key
pairs agree, and deliberately do not resolve the common label; the resulting
syndrome localizes an anomalous pair in the same way that a comparison
syndrome localizes a faulty unit.

The analogy is instructive precisely where it fails. In the classical setting
the outcomes of different comparisons are independent observations, and
enriching the test graph increases diagnosability: this is why the choice of
connection assignment is a genuine design problem. Here the corresponding lever
does not exist. Comparing all pairs rather than the adjacent ones adds nothing,
since \(G^{X}_{13}=G^{X}_{12}G^{X}_{23}\) and likewise for the \(Z\)-type
family, and more fundamentally because by Theorem~\ref{thm:storage-complete}
the whole storage-only sector has rank \(2(n-1)\) whatever comparisons are
performed. The admissible measurements form a group, so their outcomes are
algebraically dependent and the diagnostic content of the storage register is
fixed by the encoding rather than by the test structure chosen on top of it.
The same phenomenon reappears in Proposition~\ref{prop:minimality}: five checks
are informationally sufficient to separate the fault set, yet no group of rank
five achieves it.

\paragraph{Partial identification as a design target.}
Classical diagnosis also anticipates the idea that pinpointing the fault need
not be the objective. The notion of \(t/s\)-diagnosability, and more generally
the graded measures of diagnosability introduced by Barsi, Grandoni and
Maestrini~\cite{barsi1976theory}, ask only that the faulty units be confined to
a bounded set. \textit{Redemption-oriented sufficiency}
(Definition~\ref{def:redemption-resolvable}) \textit{is a refinement of that idea in
which the admissible coarseness is fixed not by the cardinality of the residual
class but by the existence of a common safe action}, which is the operationally
meaningful criterion once alternative redemption paths exist.

\paragraph{Verification in quantum secret sharing.}
The closest quantum precedent is verifiable quantum secret
sharing~\cite{crepeau2002secure,lipinska2020verifiable}, where shareholders
test the consistency of their shares against the code structure without
learning the secret. The mechanism is a relative of the one used here, but it
answers a different question. \textit{Verification is adversarial, interactive and
placed at distribution time}: it establishes that the dealer prepared a valid
encoding and that no coalition of cheating parties can corrupt the outcome.
\textit{Relational diagnosis assumes a correct encoding and asks, after a period of custody and under a noise model rather than an adversary, which of the
surviving qualified sets remains trustworthy}. The objective is accordingly
selection among alternatives rather than validation of one, the procedure is non-interactive and one-shot, and the checks are non-destructive on the code space: they project onto an eigenspace in which a healthy encoded resource
already lies, whereas verification protocols typically consume randomization
and test shares.

\paragraph{Error location as an intermediate concept.}
The coding-theoretic counterpart is older still. \textit{Error-locating
codes}~\cite{wolf1963errorlocating} were introduced as a concept intermediate
between \textit{error detection} and \textit{error correction}: the codeword is divided into sub-blocks, and decoding identifies which sub-block contains the error without correcting it. This is precisely the middle level of the hierarchy of Section~\ref{subsec:diagnostic-hierarchy}, whose three levels --- protocol consistency, inter-component diagnosis, fault localization --- correspond approximately to detection, location, and correction. Path exclusion is the QECL instance of the intermediate case, with the signal--key pairs playing the role of sub-blocks and the redemption structure supplying the reason why locating a faulty block can be enough. The correspondence is not exact, and the discrepancy is in QECL's favour: an error-locating code identifies the sub-block but not the error, whereas the inter-pair checks return the
Pauli-error class along with the pair.

\paragraph{From errors to erasures.}
A further classical fact gives a quantitative reason to invest in relational
diagnosis even when correction is unavoidable. A code of distance \(d\)
corrects \(d-1\) erasures but only \(\lfloor (d-1)/2\rfloor\) errors, \textit{because a
located error is cheaper to repair than an unlocated one}. The same relation
holds for quantum codes. A Level-3 relational syndrome that localizes a fault
to a specific component therefore does not merely inform the choice of a
redemption path: \textit{it converts an error into an erasure} for the local
error-correction layer discussed in Section~\ref{subsec:local-qec-integration}, \textit{effectively doubling the number of faults that layer can absorb}.

\paragraph{What does not transfer.}
Two classical assumptions have no counterpart here and should not be imported.
First, system-level diagnosis in the PMC tradition allows faulty testers to
return arbitrary verdicts, which is what makes \(t\)-diagnosability with
\(t>1\) delicate; the present analysis assumes ideal syndrome extraction, and
relaxing that assumption is listed among the limitations in
Section~\ref{subsec:limitations}. Second, the classical units are independent
subsystems whose outputs can be inspected individually, whereas the objects
compared here are coherent relations whose individual resolution would disturb
the protected state. The comparison model is thus the right analogy for the
structure of the syndrome, not for the physics of its extraction.

\subsection{Diagnosis, Path Exclusion, and Active Correction}
\label{subsec:diagnosis-exclusion-correction}

The operational point established in Sections~\ref{subsec:syndrome-redemption-decisions} and~\ref{subsec:redemption-sufficiency} is that diagnostic resolution should be judged by the actions it enables. Physically distinct faults need not be separated when they admit the same safe redemption action, whereas a fine syndrome remains insufficient if one residual class requires incompatible responses. Relational diagnosis therefore supports path exclusion, active correction, or acquisition of further information according to the redemption structure and the resolution available at recovery time.

\subsection{Integration with Local Quantum Error Correction and Logical Lifting}
\label{subsec:local-qec-integration}
\label{subsec:qec-logical-lifting}

Relational diagnosis is complementary to, rather than a replacement for, local quantum error correction. A QEC layer protects individual physical resources or local logical blocks and produces syndromes describing errors within those blocks. QECL relational checks operate at a different level: they test consistency relations imposed by the encrypted-cloning structure across components that may belong to distinct protected blocks or distinct redemption paths. The two sources of information can therefore be combined, with local QEC providing component-level health information and relational measurements providing protocol-level consistency information.

The Pauli-string structure of the checks considered in this work also admits a natural logical lifting. Suppose that each physical component participating in a relational observable
\[
G=\bigotimes_j P_j
\]
is itself encoded in a quantum error-correcting code. Replacing each physical Pauli \(P_j\) by the corresponding logical operator \(\overline{P}_j\) gives
\[
\overline{G}
=
\bigotimes_j \overline{P}_j.
\]
Provided that the logical operators preserve the same commutation relations with the relevant encoded fault classes, the algebraic structure of the relational syndrome is retained. Thus the prototype parity checks, the Yamaguchi--Kempf inter-pair checks, and the source-containing transverse checks can in principle be interpreted directly at the logical level.

This layered view is operationally useful because local QEC and QECL diagnosis answer different questions. \textit{A local syndrome may indicate whether a particular stored block has experienced a correctable physical error}, whereas \textit{a relational syndrome assesses whether the collection of blocks still satisfies the structure required by the QECL encoding and its redemption paths}. Information from the two levels may also help resolve residual ambiguities that neither layer eliminates independently.

A complete treatment of this integration lies beyond the scope of the present work. In particular, the choice of local codes, code distances, syndrome-extraction schedules, decoders, logical-Pauli measurement procedures, and fault-tolerant implementations introduces a separate optimization problem. The present analysis establishes only the logical compatibility of relational diagnosis with locally encoded QECL resources; the systematic design of QEC strategies adapted to the structural roles of clones and keys is left for future work.

\subsection{Implementation Outlook for Distributed Relational Measurements}
\label{subsec:distributed-implementation}

At the logical level, the diagnostic observables considered in this work are
Pauli strings and can therefore be measured using standard ancilla-assisted
parity-extraction techniques. More generally, teleportation-based constructions
allow quantum gates and joint operations to be implemented between systems that
do not interact directly, using shared entanglement, local operations, and
classical communication~\cite{gottesman1999teleportation,
eisert2000nonlocal}. These ideas form part of the broader framework of
distributed and modular quantum computation, in which nonlocal operations are
compiled into local processing supplemented by networked entanglement
resources~\cite{cirac1999distributed,jiang2007distributed,
monroe2014modular}.

The main implementation challenge arises when the qubits participating in a
relational observable are physically distributed. In such settings, joint
parity or stabilizer measurements can be mediated by distributed entangled
resource states rather than by a single ancilla interacting sequentially with
all data qubits. Related constructions use Bell or multipartite GHZ resources
to realize nonlocal couplings and distributed stabilizer extraction across
separate processing modules~\cite{jiang2007distributed,
nickerson2013topological,singh2026modular}.

These costs depend strongly on the support and weight of the diagnostic observable. The prototype checks involve only pairs of signal components, whereas the Yamaguchi--Kempf inter-pair observables act on four qubits distributed across two signal--key pairs. Source-containing checks add the further requirement that the transformed source register \(A\) remain accessible and participate in the distributed measurement. Diagnostic resolution therefore has a physical cost that is not captured by the syndrome structure alone.

The implementation problem is also coupled to the distribution architecture. Components may reside in distinct failure domains, may be connected by heterogeneous quantum links, and may already be protected locally by QEC. The most appropriate realization of a relational check can consequently depend on whether entanglement is pre-shared, whether remote gates are native, which communication latencies are tolerable, and how syndrome extraction is coordinated with local protection and storage operations.
The resulting design space is part of the broader problem of distributed
quantum computation and has been studied from computational, architectural,
and communication perspectives~\cite{caleffi2024distributed}.

A complete resource-aware treatment of distributed relational measurements lies beyond the scope of the present work. In particular, we do not optimize entanglement consumption, communication rounds, ancilla preparation, gate depth, measurement scheduling, or fault-tolerant extraction circuits. The purpose here is only to establish that the proposed diagnostic observables admit standard Pauli-parity realizations in principle, while recognizing that their efficient compilation over a distributed QECL architecture constitutes a separate design problem.

\subsection{Limitations and Open Problems}
\label{subsec:limitations}
\label{subsec:limitations-open-problems}

The present analysis is intentionally restricted to the logical structure of
\textit{one-shot relational diagnosis}. Several assumptions that make the diagnostic
mechanisms transparent will need to be relaxed before the approach can be
assessed under realistic operating conditions.

First, the explicit syndrome analyses have focused primarily on single-qubit
Pauli faults and, in particular, on the \textit{single-fault regime}. This
assumption is not purely formal: when encrypted clones are deployed at
physically separated sites and exposed to largely independent local failure
domains, isolated faults are a natural leading-order model, while simultaneous
independent faults become progressively less likely as the redundancy is
distributed. This regime is sufficient to expose the diagnostic structure of
the prototype and of the canonical Yamaguchi--Kempf protocol, but it does not
characterize correlated, multi-component, or common-mode errors. The latter
remain especially relevant for relational diagnosis because a collective
transformation may preserve all tested consistency relations while still
affecting redemption. Extending the fault model and determining which
correlated errors remain distinguishable is therefore a natural next step.

Second, the diagnostic observables identified here should not be interpreted
as exhaustive. For the Yamaguchi--Kempf protocol we derived natural check
families from the repeated Bell-label structure and, for odd clone
multiplicity, from state-independent source-containing Pauli sectors. The existence question, however, is settled: by the identification of
deterministic state-blind observables with stabilizer elements
(Section~\ref{subsec:oneshot-references}), the admissible relations of a
Clifford QECL encoding can be enumerated exhaustively. What remains open is the
harder design problem of selecting, within that group, a subgroup that is
simultaneously of low weight, compatible with the custody architecture, and
sufficient for the intended redemption policy.

Third, diagnostic capability depends on which portions of the encoded resource
remain accessible. Inter-pair checks involve only storage components, whereas
the additional odd-\(n\) refinement considered in Section~\ref{sec:input-diagnostics}
requires access to the transformed source register \(A\). \textit{Relational
diagnosability is therefore not only a property of the abstract encoding, but
also of the custody architecture} and of the stage of the protocol at which
diagnosis is performed.

\textit{The measurements themselves have also been treated ideally}. We assumed that
the intended Pauli-string syndrome can be extracted without introducing
additional faults and without uncertainty in the measurement outcome. In a
fault-tolerant distributed realization, ancilla errors, imperfect entanglement,
communication failures, and faulty syndrome extraction can themselves corrupt
the diagnostic information. Incorporating such faults requires combining the
relational layer developed here with the local-QEC and distributed-measurement
issues outlined in Sections~\ref{subsec:qec-logical-lifting}
and~\ref{subsec:distributed-implementation}.

A related restriction concerns the origin of a violated relation rather than
its extraction. The syndrome states that a prescribed relation no longer holds;
it does not state why. A relational check cannot distinguish a fault produced
by the noise processes assumed here from a deliberate manipulation of a stored
component, and the fault model adopted throughout is explicitly
non-adversarial. Extending relational diagnosis to a setting in which some
custodians may be dishonest would require combining the present construction
with the guarantees of verifiable quantum secret
sharing~\cite{crepeau2002secure}, and is a natural direction for future work.

Finally, the framework of Section~\ref{sec:relational-diagnosis} characterizes
diagnostic coverage and resolution through the syndrome partition and its
relation to redemption actions, but \textit{it does not yet provide a quantitative
optimization theory}. Useful future questions include how to select a minimal
or cost-constrained set of checks, how to compare protocols with different
redemption structures, how to weight residual ambiguities by their operational
risk, and how diagnosability should be balanced against leakage, measurement
cost, and protection overhead.

These limitations point toward a broader design problem: \textit{QECL protocols need
not merely be diagnosed after they have been chosen; their encoding and access
structures may be designed so that useful state-blind relations are
available by construction}. Developing such diagnosis-aware QECL schemes lies
beyond the scope of the present work, but follows naturally from viewing
relational diagnosability as an independent protocol property.

\section{Conclusions and Future Directions}
\label{sec:conclusions}

This work has shown that redundancy in quantum encrypted cloning can serve \textit{not only as a resource for delayed recovery, but also as a resource for one-shot fault diagnosis}. The key idea is to exploit state-independent relations imposed by the encoding rather than attempting to inspect the unknown protected state itself.

We illustrated this principle first through a prototypical duplicated-syndrome construction and then in the canonical Yamaguchi--Kempf protocol. In the latter, inter-pair Bell-label consistency checks provide pair localization and Pauli-error classification, and we proved them complete for the storage register: they generate the whole group of deterministic state-blind observables supported there, so the residual within-pair ambiguity is irreducible once the transformed source qubit is discarded. Retaining that qubit contributes exactly two further independent checks, at every clone multiplicity, and these suffice to identify every single-qubit Pauli fault. The resulting six-generator scheme for three clones, and the cheaper five-generator variant obtained by omitting the heaviest check, differ precisely in the signal--key distinction that governs the choice between path exclusion and shared-key correction. Clone parity turns out to govern the weight at which the source-assisted checks can be realized, not the resolution they achieve.

These examples motivated \textit{a general framework for one-shot relational diagnosis} based on state-blind Pauli observables, deterministic healthy references, syndrome-induced fault partitions, and redemption-oriented diagnostic sufficiency. In this perspective, relational diagnosability is a protocol property distinct from clone multiplicity, secrecy, and conventional local error-correction capability.

Several directions remain open, as discussed in Section~\ref{subsec:limitations-open-problems}. %
These include extending the analysis beyond the single-fault Pauli regime and beyond its non-adversarial reading, in which a violated relation is attributed
to noise rather than to a dishonest custodian~\cite{crepeau2002secure}, integrating relational diagnosis with fault-tolerant local QEC and distributed syndrome extraction, and developing quantitative criteria for comparing diagnostic coverage, resolution, and cost.
Since deterministic state-blind observables coincide with the stabilizer elements of the encoding, the corresponding search problem is not one of existence but of selection: which subgroup is realizable within a given custody architecture, at which measurement weight, and with which redemption-relevant resolution. A further direction is to move from the analysis of existing QECL schemes to the design of protocols whose encoding and access structures provide useful diagnostic relations by construction.

Overall, the results suggest that encrypted redundancy should be regarded as a potentially active diagnostic resource. Designing QECL systems jointly for \textit{redemption}, \textit{secrecy}, \textit{protection}, and \textit{diagnosability} may therefore provide a more complete basis for reliable long-term distributed quantum custody.

\addtocontents{toc}{\protect\setcounter{tocdepth}{1}}

\appendix

\section{Explicit Derivations}
\label{app:derivations}

Hereafter are the algebraic verifications for the canonical Yamaguchi--Kempf encoding. 
\subsection{Conventions}
\label{app:conventions}

We write \(\sigma_0=I\), \(\sigma_1=X\), \(\sigma_2=Y\), \(\sigma_3=Z\), and
\(
|\Phi^+\rangle=\frac{|00\rangle+|11\rangle}{\sqrt2},
\text{ with }
|\phi_\mu\rangle_{S_iK_i}
=
\left(\sigma_\mu^{(S_i)}\otimes I^{(K_i)}\right)|\Phi^+\rangle 
\).

The \(n\)-clone encoded state is
\[
|\Psi_{\mathrm{enc}}^{(n)}\rangle
=
\frac12
\sum_{\mu=0}^{3}
\alpha_\mu^{-1}\,
\sigma_\mu^{(A)}|\psi\rangle_A
\otimes
\bigotimes_{i=1}^{n}
|\phi_\mu\rangle_{S_iK_i},
\qquad
\alpha_0=1,\ \ \alpha_1=\alpha_3=i,\ \ \alpha_2=-i^{\,n+1}.
\]
For \(n=3\) this gives \(\alpha_0=1\), \(\alpha_1=i\), \(\alpha_2=-1\),
\(\alpha_3=i\). Since the four states \(|\phi_\mu\rangle^{\otimes n}\) are
mutually orthogonal, the encoded state is normalized for every input, and the
map \(|\psi\rangle_A\mapsto|\Psi_{\mathrm{enc}}^{(n)}\rangle\) is an isometry.

\subsection{Bell-pair eigenvalues}
\label{app:bell-eigenvalues}

Write \(B^X=X\otimes X\) and \(B^Z=Z\otimes Z\) on a single signal--key pair.
Both stabilize \(|\Phi^+\rangle\), and since \(\sigma_\mu\) either commutes or
anticommutes with \(X\) and with \(Z\),
\[
B^X|\phi_\mu\rangle=x_\mu|\phi_\mu\rangle,
\qquad
B^Z|\phi_\mu\rangle=z_\mu|\phi_\mu\rangle,
\]
with
\[
\begin{array}{c|cccc}
\mu & 0\ (I) & 1\ (X) & 2\ (Y) & 3\ (Z)\\
\hline
x_\mu & +1 & +1 & -1 & -1\\
z_\mu & +1 & -1 & -1 & +1
\end{array}
\]
The four Bell states are therefore the joint eigenstates of \(B^X\) and \(B^Z\),
and the pair \((x_\mu,z_\mu)\) is a faithful two-bit coordinate for the label
\(\mu\). The eigenvalues do not depend on the pair index, which is what makes
Proposition~\ref{prop:interpair-deterministic} work: on each branch of the
encoding all pairs carry the same \(\mu\), hence
\(G_{ij}^X\) has eigenvalue \(x_\mu^2=+1\) and \(G_{ij}^Z\) eigenvalue
\(z_\mu^2=+1\), independently of \(\mu\) and therefore of the input state.

Two further eigenvalues are used in the main text. Since
\((B^X)(B^Z)=(XZ)\otimes(XZ)=(-iY)\otimes(-iY)=-\,Y\otimes Y\), the transverse
observable \(Y\otimes Y\) on a single pair has eigenvalue \(-x_\mu z_\mu\), which
\emph{does} depend on \(\mu\); only products over an even number of pairs, or
combinations involving the source register, can be label independent.

\subsection{Relations among the source-containing checks}
\label{app:transverse-identities}

For \(n=3\) the four transverse observables of
Section~\ref{subsec:source-transverse-checks} are
\[
H_0=Y_AY_{K_1}Y_{K_2}Y_{K_3},
\quad
H_{12}=Y_AY_{S_1}Y_{S_2}Y_{K_3},
\quad
H_{13}=Y_AY_{S_1}Y_{K_2}Y_{S_3},
\quad
H_{23}=Y_AY_{K_1}Y_{S_2}Y_{S_3}.
\]
Using \(XZ=-iY\) on each qubit,
\[
G_{12}^XG_{12}^Z
=
(X_{S_1}X_{K_1}X_{S_2}X_{K_2})(Z_{S_1}Z_{K_1}Z_{S_2}Z_{K_2})
=
(-i)^4\,Y_{S_1}Y_{K_1}Y_{S_2}Y_{K_2}
=
Y_{S_1}Y_{K_1}Y_{S_2}Y_{K_2},
\]
so that, using \(Y_{K_1}^2=Y_{K_2}^2=I\),
\[
G_{12}^XG_{12}^ZH_0
=
Y_AY_{S_1}Y_{S_2}Y_{K_3}
=
H_{12},
\]
and identically \(G_{23}^XG_{23}^ZH_0=H_{23}\). For the remaining one, the
supports of \(G_{23}^X\) and \(G_{12}^Z\) overlap on \(S_2K_2\) in two
anticommuting factors, so the two commute, and
\[
G_{12}^XG_{23}^XG_{12}^ZG_{23}^ZH_0
=
Y_{S_1}Y_{K_1}Y_{S_3}Y_{K_3}\,H_0
=
Y_AY_{S_1}Y_{K_2}Y_{S_3}
=
H_{13}.
\]
Multiplying the three relations and using \(\left(G^{X}\right)^2=
\left(G^{Z}\right)^2=I\) gives \(H_0H_{12}H_{13}H_{23}=I\), so that only one of
the four is independent modulo the inter-pair generators. All signs above are
\(+1\); the minus signs appearing in earlier versions of this construction
originate in a sign misprint in the reduced state
\(\rho_{AK_1K_2K_3}\), corrected in~\cite{gianini2026full}.

\subsection{The second source-assisted generator}
\label{app:sixth-generator}

For \(n=3\) the encoding isometry maps one logical qubit into seven physical
qubits, so its image is a \([[7,1]]\) code; being Clifford, the code is a
stabilizer code and its stabilizer group has rank \(6\). The four inter-pair
checks together with \(H_0\) generate a subgroup of rank \(5\), so exactly one
further independent generator exists, in agreement with
Proposition~\ref{prop:source-assisted-rank} and
Theorem~\ref{thm:storage-complete}.

An exhaustive search over the \(4^7\) Pauli strings identifies \(64\)
deterministic state-blind observables, of which \(32\) lie outside the span of
the five checks above; \(24\) of these have the minimum weight five. Each acts on
\(A\), on both members of one distinguished pair through \(Y_SZ_K\) or
\(Z_SY_K\), and on one representative of each remaining pair. The choice adopted
in the main text is
\[
J_1=-\,X_A\,Y_{S_1}Z_{K_1}\,X_{K_2}\,X_{K_3},
\]
whose healthy eigenvalue is \(+1\) and which commutes with the other five
generators. For general odd \(n\) the analogous string
\[
J_1^{(n)}=-\,X_A\,Y_{S_1}Z_{K_1}\,X_{K_2}\cdots X_{K_n}
\]
is deterministic and independent of
\(\{G^{X}_{i,i+1},G^{Z}_{i,i+1}\}\cup\{H_0^{(n)}\}\); for even \(n\) no string of
this particular form is deterministic, and the two source-assisted generators
must instead be taken at weight \(n+2\) with one distinguished full pair, as
described in Section~\ref{sec:odd-even-diagnostics}.

\section*{Acknowledgments}
The authors used generative-AI tools (ChatGPT Sol 5.6 and Claude Opus 5) for language assistance and algebraic cross-checks. All derivations and numerical values were independently re-derived and verified by the authors, who reviewed the citations and the final manuscript text and assume full responsibility for the work. 


\bibliographystyle{unsrt}
\bibliography{refs_updated}

\end{document}